\documentclass{article}

\usepackage{fullpage,comment, amsmath, amsthm, amssymb, cleveref, url}
\usepackage[T1]{fontenc}
\usepackage[utf8]{inputenc}
\usepackage{footnote}
\usepackage{xcolor}
\usepackage{thmtools}
\usepackage{thm-restate}
\usepackage{cleveref}
\usepackage{xpatch}

\makeatletter
\xpatchcmd{\thmt@restatable}
   {\csname #2\@xa\endcsname\ifx\@nx#1\@nx\else[{#1}]\fi}
   {\ifthmt@thisistheone
       \csname #2\@xa\endcsname\ifx\@nx#1\@nx\else[{#1}]\fi%
       \else\fi}
   {}{\typeout{FIRST PATCH TO THM RESTATE FAILED}} 
\xpatchcmd{\thmt@restatable}
   {\csname end#2\endcsname}
   {\ifthmt@thisistheone\csname end#2\endcsname\else\fi}%
   {}{\typeout{FAILED SECOND THMT RESTATE PATCH}}
\makeatother

\newcommand{\ppp}{\mathbb{P}}

\DeclareMathOperator{\sign}{sign}

\newcommand{\eps}{\varepsilon}
\newcommand{\E}{\mathbb{E}}

\newcommand{\Norm}{\mathcal{N}}
\newcommand{\norm}[1]{\left\lVert#1\right\rVert^2}

\newcommand{\cphdii}{}

\renewcommand{\log}{\lg}

\newcommand{\clbi}{c_7}
\newcommand{\clbii}{c_4}
\newcommand{\clbiii}{c_5}
\newcommand{\clbiiii}{\min\left\{1/50,\clbxii/\log_2\left(e\right)\right\}}
\newcommand{\clbiiiii}{k}
\newcommand{\clbiiiiii}{l}

\newcommand{\clbxii}{c_4}
\newcommand{\clbxiii}{c_6}
\newcommand{\cubi}{2c_1^3e^8}
\newcommand{\cubii}{c_1}
\newcommand{\cubiii}{c_1}
\newcommand{\cubiiii}{c_2}
\newcommand{\cubiiiii}{c_1}

\newcommand{\R}{\mathcal{R}}
\newcommand{\Reals}{\mathbb{R}}

\newtheorem{innercustomlem}{Restatement of Lemma}
\newenvironment{customlem}[1]
  {\renewcommand\theinnercustomlem{#1}\innercustomlem}
  {\endinnercustomlem}

\newcommand{\alert}[1]{\textbf{\color{green}
[#1]}\marginpar{\textbf{\color{green}**}}\typeout{ALERT:
\the\inputlineno: #1}}

\renewcommand{\Pr}{\ppp}

\title{The Fast Johnson-Lindenstrauss Transform is Even Faster}
\author{Ora Nova Fandina \and Mikael M\o ller H\o gsgaard \and Kasper Green Larsen}

\begin{document}
 
\date{}
\maketitle

\begin{abstract}
  The seminal Fast Johnson-Lindenstrauss (Fast JL) transform by Ailon and Chazelle (SICOMP'09) embeds a set of $n$ points in $d$-dimensional Euclidean space into optimal $k=O(\eps^{-2} \ln n)$ dimensions, while preserving all pairwise distances to within a factor $(1 \pm \eps)$. The Fast JL transform supports computing the embedding of a data point in $O(d \ln d +k \ln^2 n)$ time, where the $d \ln d$ term comes from multiplication with a $d \times d$ Hadamard matrix and the $k \ln^2 n$ term comes from multiplication with a sparse $k \times d$ matrix. Despite the Fast JL transform being more than a decade old, it is one of the fastest dimensionality reduction techniques for many tradeoffs between $\eps, d$ and $n$.

In this work, we give a surprising new analysis of the Fast JL transform, showing that the $k \ln^2 n$ term in the embedding time can be improved to $(k \ln^2 n)/\alpha$ for an $\alpha = \Omega(\min\{\eps^{-1}\ln(1/\eps), \ln n\})$. The improvement follows by using an even sparser matrix. We also complement our improved analysis with a lower bound showing that our new analysis is in fact tight.
\end{abstract}
\thispagestyle{empty}
\newpage
\setcounter{page}{1}

\section{Introduction}
Dimensionality reduction is a central technique for speeding up algorithms and reducing the memory footprint of large data sets. The basic idea is to map a set $X \subset \R^d$ of $n$ high-dimensional points to a lower dimensional representation, while approximately preserving similarities between the points. The most fundamental result in dimensionality reduction, is the Johnson-Lindenstrauss transform~\cite{johnson1984extensions}, which for any precision $0 < \eps < 1$, gives a mapping $f : X \to \Reals^k$ with $k = O(\eps^{-2} \ln n)$ such that
\begin{gather}
      \label{eq:jl}
  \forall x,y \in X: \|f(x) - f(y)\|_2 \in (1 \pm \eps)\|x - y\|_2.
\end{gather}
That is, the pairwise Euclidean distance between the embeddings of any two points $x, y \in X$ is within a factor $(1 \pm \eps)$ of the original distance. The target dimensionality of $k = O(\eps^{-2} \ln n)$ is known to be optimal~\cite{DBLP:conf/focs/LarsenN17,DBLP:conf/focs/AlonK17}. For algorithmic applications where one can tolerate a small loss of precision, one can apply a Johnson-Lindenstrauss transform as a preprocessing step to reduce the dimensionality of the input. Since the running time of most algorithms depend on the dimensionality of the input, this typically speeds up the analysis while also reducing memory consumption.

A simple construction of a mapping $f$ satisfying \cref{eq:jl} is to let $f(x) = k^{-1/2}Ax$, where $A$ is a random $k \times d$ matrix, having each entry i.i.d. $\Norm(0,1)$ distributed~\cite{IM98}. This results in an embedding time of $O(kd)$ to compute the matrix-vector product $Ax$. For some applications, this embedding time may dominate the running time of the algorithms applied to the embedded data, hence dimensionality reducing maps with a faster embedding time has been the focus of much research. The line of research on faster dimensionality reducing maps splits roughly into two categories: 1) maps based on sparse matrices, and 2), maps based on structured matrices with fast matrix-vector multiplication algorithms.

\paragraph{Sparse JL.} A sparse JL transform is obtained by replacing the dense matrix $A$ above with a matrix having only $t$ non-zero entries per column. Computing the product $Ax$ now takes only $O(td)$ time instead of $O(kd)$. Perhaps even more importantly, if the input vectors $x \in X$ are themselves sparse vectors, then the embedding time is further reduced to $O(t\|x\|_0)$, where $\|x\|_0$ denotes the number of non-zero entries in $x$. This is particularly useful when applying JL on e.g. bag-of-words, $n$-gram or tf-idf representations of text documents \cite{MS99NLP}, which are often very sparse. The fastest (sparsest) known construction, due to Kane and Nelson~\cite{DBLP:journals/jacm/KaneN14}, achieves $t = O(\eps^{-1} \ln n)$, which nearly matches a sparsity lower bound by Nelson and Nguyen~\cite{DBLP:conf/stoc/NelsonN13}, stating that any Sparse JL must have $t = \Omega(\eps^{-1} \ln n/\ln(1/\eps))$. Sparse JL thus improves over classic JL by an $\eps^{-1}$ factor.

While the lower bound by Nelson and Nguyen rules out significant further improvements, the Feature Hashing technique by Weinberger et al.~\cite{WeinbergerDLSA09} study the extreme case of $t=1$. Since this is below the sparsity lower bound, they have to assume that the ratio $\nu = \|z\|_\infty/\|z\|_2$ is small for all pairwise difference vectors $z = y-x$ for $x,y \in X$ to ensure \cref{eq:jl} holds. Determining the exact ratio $\nu$ for which \cref{eq:jl} holds was subsequently done by Freksen et al.~\cite{FeatureHashing} and generalized to $t$-sparse embeddings for all $t \geq 1$ by Jagadeesan~\cite{DBLP:conf/nips/Jagadeesan19}.

\paragraph{Fast JL.} Ailon and Chazelle~\cite{Ailon2009TheFJ} initiated the study of JL transforms that exploit dense matrices with fast matrix-vector multiplication algorithms. Concretely, they defined the Fast JL transform where the embedding of a vector $x$ is computed as $PHDx$, such that $D$ is a diagonal matrix with random signs on the diagonal, $H$ is a $d \times d$ standardized Hadamard matrix and $P$ is a sparse $k \times d$ matrix. Computing $Dx$ takes only $O(d)$ time, and multiplication with the Hadamard matrix can be done in $O(d \ln d)$ time. The key observation that permits a very sparse matrix $P$, is that with high probability, the vector $y=HDx$ has a small ratio $\nu = \|y\|_\infty/\|y\|_2$, i.e. no single entry contributes most of the "mass". As was the case for Feature Hashing, such a bound allows for an even sparser random projection matrix $P$ than what a  Sparse JL transform could achieve. Ailon and Chazelle proved that a matrix $P$ in which each entry is non-zero only with probability $q = O((\ln^2 n)/d)$ suffices for \cref{eq:jl}. Thus the expected number of non-zeroes in $P$ is $kdq = O(k \ln^2 n)$ (also with high probability) and the product $Py$ can be computed in $O(k \ln^2 n)$ time. This yields a total embedding time of $O(d \ln d + k \ln^2 n)$.

Numerous follow-up works have attempted to improve over the Fast JL construction of Ailon and Chazelle, in particular attempting to shave off the $k \ln^2 n$ additive term to obtain a clean $O(d \ln d)$ time embedding. These approaches naturally divide into a couple of categories. First, a number of constructions sacrifice the optimal target dimensionality of $k = O(\eps^{-2} \ln n)$ for faster embedding time. This includes e.g. five solutions with $O(d \ln d)$ embedding time, but different sub-optimal $k=O(\eps^{-2} \ln n \ln^4 d)$~\cite{DBLP:journals/corr/abs-1009-07441}, $k=O(\eps^{-2} \ln^3 n)$~\cite{4959960}, $k=O(\eps^{-1} \ln^{3/2} n \ln^{3/2}d  + \eps^{-2} \ln n \ln^4 d)$~\cite{DBLP:journals/corr/abs-1009-07441}, $k = O(\eps^{-2} \ln^2 n)$~\cite{https://doi.org/10.1002/rsa.20360,article1,DBLP:journals/algorithmica/FreksenL20} and $k=O(\eps^{-2} \ln n \ln^2(\ln n) \ln^3 d)$~\cite{DBLP:journals/corr/abs-2003-10069}, respectively. The second category is solutions where one assumes that $k$ is significantly smaller than $d$. Here there are two solutions that both achieve $O(d \ln k)$ embedding time under the assumption that $k = o(d^{1/2})$~\cite{inproceedings,Bamberger2017OptimalFJ}. Among solutions that insist on optimal $k = O(\eps^{-2} \ln n)$ and that make no assumption about the relationship between $k$ and $d$ (other than the obvious $k \leq d$), only the recent analysis~\cite{DBLP:journals/corr/abs-2003-10069} of the Kac JL transform~\cite{kac} improves over the classic Fast JL solution by Ailon and Chazelle for some tradeoffs between $\eps, d$ and $n$. The Kac JL transform works by repeatedly picking two coordinates and doing a random unitary rotation on the two coordinates. After a sufficient number of steps, one projects on to the first $k=O(\eps^{-2} \ln n)$ coordinates and scales the coordinates appropriately. Since each rotation takes $O(1)$ time, the running time is proportional to the number of steps needed. Jain et al.~\cite{DBLP:journals/corr/abs-2003-10069} showed that
\begin{gather}
O(d \ln d + \min\{d \ln n, k \ln n \ln^2(\ln n) \ln^3 d\})\label{eq:kac}
\end{gather}
rotations suffice. Compared to the $O(d \ln d + k \ln^2 n)$ embedding time of Fast JL, Kac JL is an improvement unless $\ln^3 d > \ln n/\ln^2(\ln n)$. Despite these numerous approaches to Fast JL, we still lack a clean $O(d \ln d)$ or $O(d \ln k)$ time solution.

\paragraph{Our Contributions.}
While Fast JL has been the focus of a considerable amount of research, we give a surprising new analysis of the classic Fast JL transform by Ailon and Chazelle~\cite{Ailon2009TheFJ}. Our analysis shows that the sparsity parameter $q$ in the matrix $P$ can be lowered by a factor $\Omega(\min\{\eps^{-1} \ln(1/\eps), \ln n\})$, thereby yielding a similar improvement in embedding time. Concretely, we show that Fast JL can embed a vector $x$ in time:
\begin{gather}
\label{eq:our}
O\left(d \ln d + \min\left\{\eps^{-1} d \ln n, k \ln n \cdot \max\left\{1, \frac{\eps \ln n}{\ln(1/\eps)}\right\} \right\}\right).
\end{gather}
While this rather complicated expression might seem like an artifact of our proof, we complement our improved upper bound by showing the existence of a vector requiring precisely this embedding time using the $PHDx$ Fast JL construction. In later sections, we also give an intuitive description of where the different terms originate from.

Before giving more details on our results, let us thoroughly compare the bound to previous work. Compared to the classic $O(d \ln d + k \ln^2 n)$ Fast JL bound, we observe that \cref{eq:our} is always bounded by $O(d \ln d + k \ln n \max\{1, \eps \ln n/\ln(1/\eps)\})$, i.e. the term $O(k \ln^2 n)$ is improved by a factor $\Omega(\min\{\eps^{-1} \ln(1/\eps), \ln n\})$. Also, if we consider the case of $\eps = O(\ln(\ln n)/\ln n)$, then $1$ takes the maximum value in the $\max$-expression and the bound simplifies to $O(d \ln d + k \ln n)$. Comparing this clean bound to the Kac JL bound in~\cref{eq:kac}, this is a strict improvement (for $\eps < \ln(\ln n)/\ln n$).

In the next section, we give a detailed description of the Fast JL transform and formally state our new results.
\section{The Fast Johnson-Lindenstrauss Transform}
In the spirit of \cite{Ailon2009TheFJ} we now introduce the notation for the Fast JL transform. Here we let $d$ denote the input dimension and $k$ the output dimension. We assume $d$ is a power of two, which can always be ensured by padding with $0$'s. The Fast JL transform is the composition of three matrices $P\in\mathbb{R}^{k \times d}$ and $H,D \in \mathbb{R}^{d \times d}$. Here $D$ is a random diagonal matrix with independent Rademacher variables ($D_{i,i}$ is $1$ or $-1$ with equal probability) on its diagonal, $H$ is the normalized $d \times d$ Hadamard matrix (denoted $H_d$ in the following):
\begin{eqnarray*}
H_2 &=& \frac{1}{\sqrt{2}}\begin{pmatrix}
1 & 1 \\
1 & -1
\end{pmatrix},\\
H_d &=& \frac{1}{\sqrt{2}}
\begin{pmatrix}
H_{d/2} & H_{d/2} \\
H_{d/2} & -H_{d/2}
\end{pmatrix}
\end{eqnarray*}
and $P$ is a random matrix with the $(i,j)$'th entry being $\sqrt{1/q}\ b_{i,j}N_{i,j}$ where $b_{i,j}$ is a Bernoulli random variable with success probability/sparsity parameter $q$ and $N_{i,j}$ a standard normal random variable, where all the $b_{i,j}$'s, $N_{i,j}$'s and $D_{i,i}$'s are independent of each other.  The final embedding of a vector $x$ is then computed as $k^{-1/2}PHDx$.

\paragraph{Analysis Sketch.}
As is standard in the analysis of JL transforms, we observe that $k^{-1/2}PHD$ is a linear transformation. Hence for $k^{-1/2}PHD$ to satisfy~\cref{eq:jl} for a set of points $X$, it suffices that $k^{-1/2}PHD$ preserves the norm of every vector $z = x-y$ with $x,y \in X$ to within a factor $(1 \pm \eps)$. Also by linearity, we guarantee this by arguing that $k^{-1/2}PHD$ preserves the norm of a fixed unit vector $x$ to within $(1 \pm \eps)$ with probability $1-\delta$ when $k=O(\eps^{-2} \log(1/\delta))$. Setting $\delta = 1/n^3$ and doing a union bound over all normalized difference vectors $z/\|z\|$ with $z=x-y$ for $x,y \in X$ ensures~\cref{eq:jl} holds with probability $1-1/n$. For shorthand, we from here on use $\| \cdot \|$ to denote the norm $\| \cdot \|_2$.

To build some intuition for the key ideas used to show that the $PHD$ construction approximately preserves the norm of a unit vector with high probability, we first observe that $H$ and $D$ are both unitary matrices, hence $HDx$ preserves the norm of any vector $x$. Moreover, if we examine a single coordinate $(HDx)_i$, then it is distributed as $d^{-1/2} \sum_j \sigma_j x_j$ for independent Rademachers $\sigma_j = \sign(H_{i,j}) D_{j,j}$. Standard tail bounds show that $(HDx)_i$ is bounded by $\sqrt{\ln(d/\delta)/d}$ in absolute value with probability $1-\delta/d$ when $x$ has unit norm. A union bound over all $d$ coordinates gives that they are all bounded by $\sqrt{\ln(d/\delta)/d}$ with probability $1-\delta$. Now that $HDx$ has only small coordinates (recall $x$ has unit norm), it suffices to use a very sparse matrix $P$, precisely as in the analysis of Feature Hashing. Recall that we will set $\delta \leq 1/n^3$ and thus the $d$ term in $\ln(d/\delta)$ is irrelevant for $d \leq n$. For simplicity, we will thus assume $d \leq n$, which is also consistent with previous work (it was assumed both for Fast JL~\cite{Ailon2009TheFJ} and Kac JL~\cite{DBLP:journals/corr/abs-2003-10069}).

\paragraph{Upper Bounds.}
In their seminal work, Ailon and Chazelle~\cite{Ailon2009TheFJ} showed that it suffices to set 
\[
q = O(\ln^2(n)/d)
\]
to guarantee~\cref{eq:jl} for a set $X$ of $n$ points (with probability $1-1/n$ by setting $\delta=1/n^3$). Their proof follows the template above, union bounding over preserving the norm of all normalized pairwise difference vectors. This results in an expected $kdq = O(k \ln^2 n)$ number of non-zero entries in $P$. Our main upper bound result is an improved analysis, showing that an even sparser $P$ suffice:

\begin{restatable}{Theorem1}{sensone}
\label{phdtheoremuppermain}
Let $X$ be a set of $n$ vectors in $\mathbb{R}^d$ and let $k=\Theta(\eps^{-2}\ln n)$. Let further $0<\varepsilon \leq C$ where $C$ is some universal constant. Then for 
\[
q=O\left(\min \left\{\eps, \frac{\ln n}{d} \cdot \max\left\{ 1,  \frac{\eps \ln n}{ \ln(1/\eps)}\right\}\right\}\right),
\]
we have that $k^{-1/2}PHD$ guarantees~\cref{eq:jl} with probability at least $1-1/n$. 
\end{restatable}

Compared to \cite{Ailon2009TheFJ} which uses $q=O(\ln^{2}(n)/d)$, we notice that even if we ignore the first term in the $\min$-expression, our guarantee on $q$ is $q = O(\max\{\ln(n)/d, \eps \ln^2(n)/(d\ln(1/\eps)))$, i.e. always at least a factor $\Omega(\min\{\ln n,\eps^{-1} \ln(1/\eps)\})$ better. Also, for the case of $\eps = O(\ln(\ln n )/\ln n)$, the $1$-term in the max dominates, and the expression for $q$ simplifies to a clean $q = O(\ln(n)/d)$. Plugging in the value of $q$ from Theorem~\ref{phdtheoremuppermain} (and recalling $k = \Theta(\eps^{-2} \ln n)$), we get that the number of non-zeroes of $P$ is 
\[
kdq = O\left(\min \left\{\eps^{-1} d \ln n, k \ln n \cdot \max\left\{ 1,  \frac{\eps \ln n}{ \ln(1/\eps)}\right\}\right\}\right),
\]
in expectation. Moreover, since this number is larger than $\ln n$, it follows from a Chernoff bound that the number of non-zeroes is strongly concentrated around its mean. 

\paragraph{Lower Bound.}
A natural question to ask now is whether the above $q$ is optimal, or an even more refined analysis can lead to further improvements. To answer this question, we show an example of a unit vector $x$, such that for the mapping $k^{-1/2}PHDx$ to preserve the norm of $x$ to within $(1\pm \eps)$ with probability $1-\delta$, we cannot make $P$ sparser than in Theorem~\ref{phdtheoremuppermain}:
\begin{restatable}{Theorem1}{senstwo}
\label{phdtheorem2main}
For $0<\delta,\varepsilon\leq r$ where $r$ is a universal constant and $k=\eps^{-2}\ln(1/\delta)$, there is a unit vector $x \in \mathbb{R}^d$ for which we must have 
\[
q=\Omega\left(\min \left\{\eps, \frac{\ln(1/\delta)}{d} \cdot \max\left\{ 1,  \frac{\eps \ln(1/\delta)}{ \ln(1/\eps)}\right\}\right\}\right),
\]
for
\begin{gather*}
\frac{1}{\sqrt{k}}\|PHDx\| \in (1\pm\varepsilon),
\end{gather*}
to hold with probability at least $1-\delta$.
\end{restatable}
For the reader concerned with assuming $k=\eps^{-2} \ln(1/\delta)$, we remark that \Cref{phdtheorem2main} can also be shown with $k=\tilde{c}\eps^{-2}\ln(1/\delta)$ for $\tilde{c}\geq1$, and another universal constant $r'$. 

Comparing \Cref{phdtheorem2main} to \Cref{phdtheoremuppermain}, we observe that the bound on $q$ match exactly when setting $\delta = n^{-\Theta(1)}$. This means that the analysis of Fast JL cannot be improved if one attempts to show that any fixed vector has its norm preserved except with probability $n^{-\Theta(1)}$ and doing a union bound over all pairwise difference vectors. It is however still conceivable that a more refined analysis could somehow argue that there are only very few worst case vectors in any set $X$. However, such an improved analysis remains to be seen for any JL transform (when focusing only on the type of guarantee in~\cref{eq:jl}, whereas net-based arguments have been used e.g. for subspace embeddings~\cite{DBLP:conf/stoc/ClarksonW13}). In this light, \Cref{phdtheorem2main} can be seen either as a hard barrier for Fast JL, or as hinting at a way towards further improvements.

In the next section, we formally prove \Cref{phdtheoremuppermain} and also discuss how our analysis differs from the previous analysis by Ailon and Chazelle and conclude by giving more intuition on where the different terms in the expression for $q$ come from.
\section{Upper Bound}
In this section we give the proof of \Cref{phdtheoremuppermain}. We start by giving the high level ideas of our proof. As in previous works, our analysis follows by arguing that for any fixed unit vector $x$, it holds with probability at least $1-1/n^3$ that $\|k^{-1/2}PHDx\| \in (1 \pm \eps)$. 

First, we observe that $HD$ is a unitary matrix and thus $\|HDx\|=\|x\|=1$ for a unit vector $x$. Moreover, any single coordinate $(HDx)_i$ equals $d^{-1/2} \sum_{j=1}^d \sigma_j x_j$, where the $\sigma_j = D_{j,j}\sign(H_{i,j})$'s are independent Rademacher random variables. Thus in line with the analysis by Ailon and Chazelle~\cite{Ailon2009TheFJ}, we get that any coordinate $(HDx)_i$ is bounded by $O(\sqrt{\ln(n)/d})$ in absolute value with probability $1-1/n^4$. A union bound over all $d \leq n$ coordinates (this assumption is also made in previous work) gives that all coordinates of $HDx$ are bounded by $O(\sqrt{\ln(n)/d})$ with probability $1-1/n^3$.

What remains now is to argue that $k^{-1/2}\|Pu\| \in (1 \pm \eps)$ with high probability when $u = HDx$ is a unit vector with all coordinates bounded by $O(\sqrt{\ln(n)/d})$. 

To simplify the analysis, we will argue that $k^{-1}\|Pu\|^2 \in (1 \pm \eps)$ with probability $1-1/n^3$. This is stronger since $\sqrt{1\pm\eps} \subset (1\pm\eps)$. To understand the distribution of $\|Pu\|^2$ for a fixed $u$, notice that the $i$'th coordinate of $Pu$ is given by $\sum_{j=1}^d q^{-1/2} u_j b_{i,j} N_{i,j}$ by definition of $P$. Let us assume that the Bernoulli random variables $b_{i,j}$ have been fixed. In this case, $(Pu)_i$ is a sum of weighted and independent $\Norm(0,1)$ random variables. Hence $(Pu)_i$ is itself $\Norm(0,q^{-1} \sum_{j=1}^d b_{i,j} u_j^2)$ distributed. Now define $Z_i = \sum_{j=1}^d b_{i,j} u_j^2$ and let $N_1,\dots,N_k$ be independent $\Norm(0,1)$ random variables. We see that, for fixed values of all Bernoullis, $\|Pu\|^2$ is distributed as $\sum_{i=1}^k q^{-1} (\sqrt{Z_i} N_i)^2$, which is equal to $\sum_{i=1}^k q^{-1} Z_i N_i^2$. Our proof now has two steps: 1.) Give a bound on the $Z_i$'s that holds with high probability over the random choice of the Bernoullis $b_{i,j}$, and 2.), use the bound on the $Z_i$'s to argue that $\sum_{i=1}^k q^{-1} Z_i N_i^2$ behaves in a desirable manner.

In order to understand what type of bounds we need on the $Z_i$'s, we start by examining step 2. For this step, we need a tail bound on $\sum_{i=1}^k q^{-1} Z_i N_i^2$. When the $Z_i$'s are fixed, this is a weighted sum of sub-exponential random variables. To analyse it, we use Proposition 5.16 from \cite{DBLP:books/cu/12/VershyninEK12}, which gives upper bounds on the tails of centered sub-exponential random variables:

\begin{restatable}[\cite{DBLP:books/cu/12/VershyninEK12}]{Lemma1}{sensthree}
\label{chisquareuppertails1} Let $Y_{1}, \ldots, Y_{k}$ be independent centred sub-exponential random variables in the sense that there exist a constant $C>0$ such that $\mathbb{E} [\exp \left(C Y_{i}\right)] \leq e$. Then for every $a_{1}, \ldots, a_{k}\in \mathbb{R}$ and  $R= a_{1} Y_{1}+\cdots+a_{k} Y_{k}$ we have 
$$
\begin{array}{ll}
\mathbb{P}\left[ |R| \geq x\right]  \leq 2 \exp \left(-\frac{c x^{2}}{\|a\|_{2}^{2}}\right), & \forall 0 \leq x \leq \frac{\|a\|_{2}^{2}}{\|a\|_{\infty}} \\
\mathbb{P}[ |R| \geq x]  \leq 2 \exp \left(-\frac{c x}{\|a\|_{\infty}}\right), & \forall x \geq \frac{\|a\|_{2}^{2}}{\|a\|_{\infty}}.
\end{array}
$$
where $c>0$ is an absolute constant.
\end{restatable}

Note that for a random variable $N \sim \Norm(0,1)$, we have that the centred square (i.e. $N^2-1$) is a sub-exponential random variable in the spirit of \Cref{chisquareuppertails1}. This can be seen by $|t|\leq 0.3$ we have that
\begin{gather*}
\E\left[ \exp\left(t(N^2-1) \right)\right] \leq \E\left [\exp(tN^2)\right ]=\left(1-2t\right)^{-\frac{1}{2}}=\exp\left(-\frac{\ln(1-2t)}{2}\right)\leq \exp\left(\frac{(-2t)+(-2t)^2}{2}\right)\leq e,
\end{gather*}
where the first equality follows by the $\chi^2$-distribution's moment generating function and the second to last inequality follows by $-\ln(1+x)\leq x+x^2$ for $x>-0.68$. So for $C=0.3$  we can apply \Cref{chisquareuppertails1} to $\sum_{i=1}^k q^{-1} Z_i N_i^2$ by rewriting as $\sum_{i=1}^k q^{-1} Z_i (N_i^2-1) + \sum_{i=1}^k q^{-1} Z_i$. The latter term is constant when the Bernouillis have been fixed and thus we may use \Cref{chisquareuppertails1}.

Examining \Cref{chisquareuppertails1}, we see that we need two bounds on the $Z_i$'s, one on $\sum_i Z_i^2$ and one on $\max_i |Z_i|$. Thus for step 1., we focus on giving bounds on these two quantities. For this, we will use that $u = HDx$ has all coordinates bounded in absolute value by $O(\sqrt{\ln(n)/d})$ as observed earlier. We then argue that the hardest such vector $u$, is one in which precisely $m$ coordinates all take the value $m^{-1/2} = O(\sqrt{\ln(n)/d})$ and the remaining coordinates of $u$ are all $0$. This is also the hard vector analysed by Ailon and Chazelle. In their analysis, they simply bound $\sum_{i=1}^k Z_i^2$ by $k (\max_i |Z_i|)^2$ and this is where we improve over their work. Giving a tight analysis of $\sum_i Z_i^2$ is far from trivial and takes up the majority of \Cref{sec:upperin}.

For now, we merely state the concentration inequalities we need and return to proving them in \Cref{sec:upperin}. For bounding $\max_i Z_i$, we prove the following lemma:

\begin{restatable}{Lemma1}{sensfour}
\label{ublemma1}
For $i=1,\ldots,k$ let $Z_i=\sum_{j=1}^d u_j^2 b_{i,j}$ where the $b_{i,j}$'s are independent Bernoulli random variables with success probability $q$ and the $u_j^2$'s are positive real numbers bounded by $1/m$ and summing to 1. We then have for $\alpha\leq 1/4$ that
\begin{gather*}
\ppp\left[ \max_{i=1,\ldots,k} Z_i>\frac{q}{2\alpha }\right] \leq k\exp\left(-\frac{mq\ln(1/\alpha)}{32\alpha}\right).
\end{gather*}
\end{restatable}
And to bound $\sum_i Z_i^2$, we show the following:
  \begin{restatable}{Lemma1}{sensfive}
  \label{ublemma3}
    Let $Z_1,\dots,Z_k$ be i.i.d. random variables distributed as the $Z_i$'s in \Cref{ublemma1}. Then for any $t \geq 64\cdot24 e^3 q^2 k$ and $q \geq 8/(em)$, we have:
    \[
      \Pr\left[\sum_{i=1}^k Z_i^2 >  t\right] < 14\exp\left(-\frac{m\sqrt{t}\ln(\sqrt{t/2^{3}}/(eq))}{200\cdot44\cdot2^{\frac{5}{2}}} \right).
      \]
  \end{restatable}
Before continuing, let us briefly argue that \Cref{ublemma3} is tighter than using the approach of Ailon and Chazelle where $\sum_i Z_i^2$ is merely bounded as $k (\max_i Z_i)^2$. For large enough $t$, \Cref{ublemma3} roughly gives that $\ppp[ \sum_i Z_i^2 > t]  < \exp(-m \sqrt{t} \ln(\sqrt{t}/q))$. If we instead bounded $\sum_i Z_i^2$ by $k (\max_i Z_i)^2$, then for any $t$, their approach would need $\max_i Z_i \leq \sqrt{t/k}$. Choosing $\alpha$ such that $\sqrt{t/k} = q/(2\alpha)$ and examining \Cref{ublemma1}, we would roughly get $\ppp[ \sum_i Z_i^2 > t]  < k\exp(-(m (\sqrt{t/k}) \ln((\sqrt{t/k})/q)))$. We would thus lose almost a factor $\sqrt{k}$ in the exponent. This is basically where our improvement comes from.

Unfortunately, \Cref{ublemma3} does not capture all tradeoffs between $\eps, d$ and $n$ that we need. Thus we also need the following alternative to \Cref{ublemma3}:

 \begin{restatable}{Lemma1}{senssix}
 \label{ublemma4}
 Let $Z_1,\dots,Z_k$ be i.i.d. random variables distributed as the $Z_i$'s in \Cref{ublemma1}, with $m=\cubiiii d/\ln n$ and the embedding dimension $k= \cubiii \eps^{-2} \ln n$ and $q =\cubii \varepsilon$, where $\cubii\geq 1/\cubiiii$. For $\varepsilon\leq \cubii^{-1}/(e4)$ and $t\geq \cubi \ln n$, we have that 
    \[
      \Pr\left[\sum_{i=1}^k Z_i^2 >  t\right] \leq 3n^{-4\cubiiiii}.
    \]
 \end{restatable}

With the central lemmas laid out, we now give the full proof details by following the above proof outline. The proofs of \Cref{ublemma1}, \Cref{ublemma3} and \Cref{ublemma4} can be found in \Cref{sec:upperin}.

\paragraph{Proof of \Cref{phdtheoremuppermain}.}
\begin{proof}
 Let further $m=\cubiiii d/\ln n$ for a small enough constant $\cubiiii$.
Let the embedding dimension $k=\cubiii \eps^{-2} \ln n$, with $\cubiii \geq 1/\cubiiii$. Let the success probabilities of the binomial random variables $b_{i,j}$ in $P$ be 
\[
q=\max\left\{\cubiiiii/m,\cubii\varepsilon\min \left\{1,\ln \left(n\right)/\left(m\ln\left (1/\varepsilon\right)\right)\right\} \right\}.
\]
 Assume for now that $u$ is a vector in $\mathbb{R}^d$ such that $u_i^2\leq 1/m$ for all $i=1,\ldots,d$ and $\norm{u}=1$. By construction of $P$ and the $2$-stability of the standard normal distribution we have that 
\begin{gather*}
\norm{Pu}\stackrel{}{=}\sum_{i=1}^k \left(\sum_{j=1}^{d} \sqrt{1/q}u_j b_{i,j} N_{i,j}\right)^2\stackrel{d}{=}\sum_{i=1}^\clbiiiii \frac{1}{q}Z_i  N_{i}^2,
\end{gather*}
where $Z_i=\sum_{j=1}^{d}  u_j^2b_{i,j}$ and  $N_i$'s are independent standard normal random variables. 
We first prove a bound on $\sum_{i=1}^k Z_i$. For this, notice that $\sum_{i=1}^k Z_i$ is a sum of independent random variables, where each $Z_i$ is a sum of independent random variables with values between $[0,1/m]$. Furthermore, we have $\E\left[Z_i\right]=q$, implying that $\E [\sum_{i=1}^k mZ_i]=qmk$. We therefore get by a Chernoff bound that

\begin{gather*}
\ppp\left[\sum_{i=1}^k Z_i\not\in (1\pm\varepsilon/4)qk\right]=\ppp\left[ \sum_{i=1}^k mZ_i\not\in (1\pm\varepsilon/4)qmk\right] \leq 2\exp\left(-\frac{qmk\varepsilon^2}{48}\right)\leq 2n^{-\cubiii^2 /48},
\end{gather*}
where the last inequality follows by $q\geq \cubiiiii/m$ and $k=\cubiiiii \eps^{-2} \ln n$, so $qmk\varepsilon^2\geq\cubii^2\ln n$. Thus we have $\sum_{i=1}^k Z_i\in (1\pm\varepsilon/4)qk$ with probability at least $1-2n^{-\cubiii^2/48}$.

In the following we do a case analysis based on the value of $q$. Our goal is to show that $\|Pu\|^2 = \sum_i Z_i N_i^2/q \in (1 \pm \eps/4)k$ with high probability (conditioned on $u$ having bounded coordinates as remarked earlier).

\subsubsection*{Cases $q=\cubiiiii/m$ and $q= \cubii\varepsilon\ln(n)/(m\ln(1/\varepsilon))$.}
We treat the cases $q=\cubiiiii/m$ and $q= \cubii\varepsilon\ln(n)/(m\ln(1/\varepsilon))$ in a similar manner. In both these cases, we have $q\geq  \cubii\varepsilon\ln(n)/(m\ln(1/\varepsilon))$ (due to the $\max$ in the definition of $q$). Thus \Cref{ublemma1}, with $\alpha=\eps$ now implies that $\|Z\|_\infty=\max_{i=1,\ldots,k} Z_i \leq q/(2\varepsilon) $ with probability at least $1-k\exp(-(mq\ln(1/\varepsilon))/(32\varepsilon))\geq 1-n^{-\cubii/32 + 1}$ (which follows by $mq\ln (1/\varepsilon)/\varepsilon\geq \cubii \ln n$ and $k \leq n$). 

Using $q \geq c_1/m$ we may invoke \Cref{ublemma3}. Combining this with $q\geq \cubii\varepsilon\ln(n)/(m\ln(1/\varepsilon))$ we conclude that $\|Z\|^2=\sum_{i=1}^k Z_i^2\leq 64\cdot24 \cdot e^3q^2k$ with probability at least  

\begin{gather*}1-14\exp\left(-\frac{m\sqrt{64\cdot24 \cdot e^3q^2k}\ln(\sqrt{(64\cdot24 \cdot e^3q^2k)/(2^{3})}/(eq))}{200\cdot44\cdot2^{\frac{5}{2}}} \right)\\
\geq 1- 14\exp\left(-\frac{(\cubii\ln (n))^{3/2}\ln\left(22\sqrt{k}\right)}{300\ln\left(1/\eps \right) } \right)\\
\geq 1-14n^{-\cubii^{3/2}/300},
\end{gather*}
where in the first inequality we used that $(\sqrt{64\cdot24e^3})/(200\cdot44\cdot 2^{5/2})\geq 1/300$, $\sqrt{(64\cdot24e)/2^3}\geq 22$ and $mq\sqrt{k}\geq
 (c_1\varepsilon\ln(n)/\ln(1/\varepsilon))\sqrt{c_1\ln (n)/\varepsilon^2}
 =(\cubii \ln n)^{3/2}/\ln\left(1/\eps \right)$ and in the second inequality that $\ln (22\sqrt{k})/(\ln\left(1/\eps\right))=\ln(22\sqrt{c_1\ln (n)/\varepsilon^2})/(\ln\left(1/\eps \right))\geq 1$.

 Hence in these cases we have that $\sum_{i=1}^k Z_i\in (1\pm\varepsilon/4)qk$, $\|Z\|_\infty=\max_{i=1,\ldots,k} Z_i \leq \frac{q}{2\varepsilon} $ and $\|Z\|^2=\sum_{i=1}^k Z_i^2\leq 64\cdot 24 e^3q^2k$ with probability at least $1-17n^{-\cubiii/300+1}$. We call such outcomes of the variables $Z_i$ \emph{desirable}.

 We now notice that for desirable outcomes of the $Z_i$'s, we have from \Cref{chisquareuppertails1} that if $(\varepsilon/4) \sum_{i=1}^k Z_i \geq \|Z\|^2/\|Z\|_\infty$, then (with probability over the $N_i$'s)
\begin{gather*}
\ppp\left[  \sum_{i=1}^k \frac{1}{q}N_i^2Z_i\not\in (1\pm\varepsilon/4)\sum_{i=1}^k \frac{1}{q}Z_i\right] \leq 2\exp\left(-\frac{c(\varepsilon/4)\sum_{i=1}^k Z_i}{\|Z\|_\infty}\right)\leq 2\exp\left(-\frac{c\varepsilon qk/8}{q/(2\varepsilon) }\right)= 2n^{-c \cubiii/4},
\end{gather*}
where we used that $k=\cubiii \eps^{-2} \ln n$. On the other hand, if $(\varepsilon/4) \sum_{i=1}^k Z_i \leq \|Z\|^2/\|Z\|_\infty$, then by \Cref{chisquareuppertails1} (and using $\eps < 1$):
\begin{gather*}
\ppp\left[ \sum_{i=1}^k \frac{1}{q}N_i^2Z_i\not\in (1\pm\varepsilon/4)\sum_{i=1}^k \frac{1}{q}Z_i\right] \leq 2\exp\left(-\frac{c((\varepsilon/4)\sum_{i=1}^k Z_i)^2}{\|Z\|^2}\right)\\
\leq 2\exp\left(-\frac{c\varepsilon^2q^2k^2}{16 \cdot 64\cdot96e^3q^2k}\right)= 2n^{-c \cubiii/(16 \cdot 64\cdot 96 e^3)}.
\end{gather*}
By this we conclude that for desirable outcomes of the $Z_i$'s, for the constant $r_1:=c/(16 \cdot 64\cdot 96 e^3)$, it holds (with probability over the $N_i$'s):
\begin{gather*}
1-2n^{-r_1 \cubiii} \leq\ppp\left[  \sum_{i=1}^k \frac{1}{q}N_i^2Z_i\in (1\pm\varepsilon/4)\sum_{i=1}^k \frac{1}{q}Z_i\right] \leq\ppp\left [ \sum_{i=1}^k \frac{1}{q}N_i^2Z_i\in (1\pm \varepsilon)k\right ], 
\end{gather*}
where in the last inequality we used that for desirable outcomes of the $Z_i$'s it holds $\sum_{i=1}^k Z_i\in (1\pm\varepsilon/4)qk$.

Since the $Z_i$'s and $N_i$'s are independent, it follows from the above that with probability at least $(1-2n^{-r_1\cubiii})\cdot(1-17n^{-\cubii/300+1})\geq1-34n^{-\min\{r_1,1/300\}\cubiii+1}$ it holds that $\sum_{i=1}^k N_i^2Z_i/q\in (1\pm \varepsilon)k$.

\subsubsection*{Case $q=\cubii\varepsilon$.}
In the case that $q=\cubii\varepsilon$ (we assume that $\varepsilon<\cubii^{-1}/(4e)$), it follows from \Cref{ublemma4} with $t= 2\cubii^3e^8\ln n$ that $\|Z\|^2=\sum_{i=1}^k Z_i^2\leq 2\cubii^3e^8 \ln n$ with probability at least $ 1-3n^{-4\cubiiiii}$. Thus we conclude that with probability at least $1-5n^{-\cubii/48}$ we have $\sum_{i=1}^k Z_i\in (1\pm\varepsilon/4)qk$ and $\|Z\|^2=\sum_{i=1}^k Z_i^2\leq 2\cubii^3e^8 \ln n$. In this part of the case analysis, we refer to such outcomes as \emph{desirable}.

Now for desirable outcomes of the $Z_i$'s, we get again using \Cref{chisquareuppertails1} that if $(\varepsilon/4) \sum_{i=1}^k Z_i \geq \|Z\|^2/\|Z\|_\infty$, then with probability over the $N_i$'s, and using the trivial bound that the $Z_i$'s are at most 1, it follows that
\begin{gather*}
\ppp\left[  \sum_{i=1}^k \frac{1}{q}N_i^2Z_i\not\in (1\pm\varepsilon/4)\sum_{i=1}^k \frac{1}{q}Z_i\right] \leq 2\exp\left(-\frac{c(\varepsilon/4)\sum_{i=1}^k Z_i}{\|Z\|_\infty}\right)\leq 2\exp\left(-\frac{c\varepsilon qk}{8}\right)= 2n^{-c \cubiii^2/8},
\end{gather*}
where the last inequality follows from $\sum_{i=1}^k Z_i \geq (1-\eps/4)qk \geq qk/2$ and the equality follows from $\varepsilon qk=c_1^2\ln n$. In the case of $(\varepsilon/4) \sum_{i=1}^k Z_i \leq \|Z\|^2/\|Z\|_\infty$, \Cref{chisquareuppertails1} yields:
\begin{gather*}
\ppp\left[  \sum_{i=1}^k \frac{1}{q}N_i^2Z_i\not\in (1\pm\varepsilon/4)\sum_{i=1}^k \frac{1}{q}Z_i\right] \leq 2\exp\left(-\frac{c((\varepsilon/4)\sum_{i=1}^k Z_i)^2}{\|Z\|^2}\right)\\ \leq 2\exp\left(-\frac{c\varepsilon^2q^2k^2}{128\cubii^3e^8\ln n}\right)\leq 2n^{-c \cubiii/(128e^8)},
\end{gather*}
where the last inequality follows from $\varepsilon^2q^2k^2/\ln n=\cubii^4\varepsilon^4\ln^2(n) /(\varepsilon^4\ln n)\geq \cubii^4\ln\ n$.

Now, let $r_2=c/(128e^8)$. From the above, we conclude that for desirable outcomes of the $Z_i$'s, with probability (over the $N_i$'s):
\begin{gather*}
1-2n^{r_2 \cubiii} \leq\ppp\left[  \sum_{i=1}^k \frac{1}{q}N_i^2Z_i\in (1\pm\varepsilon/4)\sum_{i=1}^k \frac{1}{q}Z_i)\right] \leq\ppp\left[  \sum_{i=1}^k \frac{1}{q}N_i^2Z_i\in (1\pm \varepsilon)k\right] ,
\end{gather*}
and again using the independence of the $Z_i$'s and $N_i$'s, we get that $\sum_{i=1}^k N_i^2Z_i/q\in (1\pm \varepsilon)k$ holds with probability at least $(1-2n^{-r_2 \cubiii})(1-5n^{-\cubiii/48})\geq 1-10n^{-\min\{r_2,1/48\}\cubiii}$.

\subsubsection*{Conclusion.} 
In the above we had assumed that the vector $u$ had entries $u_i^2\leq 1/m$ and had unit length. By a similar argument to \cite{Ailon2009TheFJ} equation (4) page 308, we get that with probability at least $1-1/(2n^3)$, it holds that $u_i^2=(HDx)_i^2\leq \ln (n)/(\cubiiii d)=1/m$ for all $i=1,\ldots,d$ simultaneously, when $\cubiiii$ is small enough (assuming $d \leq n$ such that $\ln d = O(\ln n)$), thus we have $u_i^2\leq 1/m$ as required.

From the above, we see that in all cases, if we set $\cubiii$ as a sufficiently large constant, then with probability at least $1-1/(2n^3)$, we have $\sum_{i=1}^k \frac{1}{q}N_i^2Z_i\in (1\pm \varepsilon)k$. Since $\|Pu\|^2$ was equal in distribution to $\sum_{i=1}^k \frac{1}{q}N_i^2Z_i$, the same holds for $\|Pu\|^2$. 

Since $D$ is independent of $P$, we get that with probability at least $(1-1/(2n^3))(1-1/(2n^3)) \geq 1-1/n^3$, we have $k^{-1}\|PHDx\|^2 \in (1 \pm \eps)$ as desired.

For a set of $n$ vectors $X$, we finally union bound over all vectors $z/\|z\|$ where $z = x-y$ with $x,y \in X$. There are less than $n^2$ such pairs and we conclude that with probability at least $1-1/n$, we have that $k^{-1/2}PHD$ guarantees~\eqref{eq:jl}.

We now claim that our choice of 
\[
q=\max\left\{\cubiiiii/m,\cubii\varepsilon\min\left \{1,\ln (n)/(m\ln (1/\varepsilon))\right \}\right \},
\]
is equivalent to that claimed in \Cref{phdtheoremuppermain}. Recalling that $m = O(d/\ln n)$, we see that our choice of $q$ is $O(\max\{(\ln n)/d, \eps \min\{1, \ln^2(n)/(d\ln(1/\eps))\}\})$. Since $(\ln n)/d \leq (\ln n)/k = O(\eps^2) = O(\eps)$, we can never have $(\ln n)/d = \omega(\eps)$ and hence we can move the $\max$ into the min and get
\[
q=O\left(\min\left\{\eps, \frac{\ln n}{d} \cdot \max\left\{1, \frac{\eps \ln n}{\ln(1/\eps)}\right\} \right\}\right).
\]
This completes the proof of \Cref{phdtheoremuppermain}.
\end{proof}

\subsubsection*{Discussion of Expression.}
Let us conclude by giving some more intuition on where the different terms in the expression for $q$ originate from. Recall from above that the hardest vector for $k^{-1/2}P$ is a unit vector $u$ with $m = O(d/\ln n)$ non-zero entries, each of magnitude $m^{-1/2}$. Also recall that each entry of $P$ is the product of a Bernoulli $b_{i,j}$ with success probability $q$ and a normal distributed random variable with variance $1/q$.

The term $\ln(n)/d$ in the expression for $q$ intuitively comes from the following: There is a total of $km$ Bernoulli random variables $b_{i,j}$ that are each multiplied with the same non-zero value $u_j^2$. This gives an expected $kmq$ of them that are non-zero. Intuitively, since they are all multiplied with the same coefficient, we need the number of non-zero Bernouillis to be within $\eps k m q$ of the expectation. A binomial distribution with $km$ trials and success probability $q$ deviates from its expectation by $\Omega(\sqrt{kmq \ln n})$ with probability $n^{-1/2}$ and thus we require $\sqrt{kmq \ln n} < \eps k m q$. This implies that we must set $q > \ln(n)/(\eps^2mk) =\Omega(1/m) = \Omega(\ln(n)/d)$.

The terms $\eps \ln^2 n/(d \ln(1/\eps))$ and $\eps$ in the expression for $q$ come from the event that the square of the first coordinate, $(k^{-1/2}Pu)_1^2$ is larger than $\eps$ (which causes a distortion if the rest of the coordinates are concentrated). Conditioned on the Bernoullis $b_{1,j}$, the square of the first coordinate is the square of a normal distributed random variable. Hence it is a factor $\Omega(\ln n)$ larger than its variance with probability $n^{-1/2}$. There are now two cases: 1. $m < c\ln_{1/q} n$ for a small constant $c>0$, and 2., $m \geq c\ln_{1/q} n$. 

In the first case, $m < c \ln_{1/q} n$, it happens with probability at least $n^{-1/2}$ that all Bernoullis $b_{1,j}$ that are multiplied with a non-zero coefficient take the value $1$. In that case, the first coordinate of $k^{-1/2}Pu$ is normal distributed with mean zero and variance $1/(qk)$ (since $\sum_j u_j^2=1$). We thus need $\ln n/(qk) < \eps$. Using that $k = \Theta(\eps^{-2} \ln n)$, this means we have to set $q = \Omega(\eps)$. 

In the second case, $m \geq c \ln_{1/q} n$, we expect to see $qm$ non-zero Bernoullis $b_{1,j}$ that are each multiplied with $1/m$ for the first coordinate of $k^{-1/2}Pu$. However, by a "reverse" Chernoff bound, with probability at least $n^{-1/2}$, we see at least $c\ln_{1/q} n$ non-zero Bernoullis. In that case, the first coordinate of $k^{-1/2}Pu$ is normal distributed with mean zero and variance $\Theta((\ln_{1/q} n)/(mqk)) = \Theta(\eps^2 \ln_{1/q}(n)/(dq))$. Since the square of the first coordinate was a factor $\ln n$ larger than its variance with probability $n^{-1/2}$, we hence need $\eps^2 \ln n \ln_{1/q}(n)/(dq) = O(\eps)$. If we for simplicity approximate $q$ by $\eps$ in $\ln_{1/q} n$, this gives precisely $q = \Omega(\eps \ln^2 n/(d \ln(1/\eps)))$.
\section{Lower Bound}
In this section, we prove the lower bound in \Cref{phdtheorem2main}. That is, we give an example of a unit vector $x \in \mathbb{R}^d$, such that one must have
\[
q=\Omega\left(\min \left\{\eps, \frac{\ln(1/\delta)}{d} \cdot \max\left\{ 1,  \frac{\eps \ln(1/\delta)}{ \ln(1/\eps)}\right\}\right\}\right),
\]
to guarantee $\ppp[ \|k^{-1/2}PHDx\| \in (1 \pm \eps)]  \geq 1-\delta$. 

The proof of the lower bound goes in two steps. In the first step, we show that we must have $q =\Omega(\ln(1/\delta)/d)$. In the second step, we use the result from step one to conclude that $q$ must also be $\Omega(\varepsilon\min\{1,\ln^2(1/\delta)/(d\ln (1/\eps))\})$. Combining the two, we have:
\[
q = \Omega\left(\max\{\ln(1/\delta)/d, \eps \min \{1,\ln^2(1/\delta)/(d\ln (1/\eps))\}\}\right).
\]
Noticing that we always have $\ln(1/\delta)/d = O(\ln(1/\delta)/k) = O(\eps^2) = O(\eps)$, we can move the max inside the min and obtain the bound claimed above.

In both steps, we use the same hard instance vector $x$. This hard vector $x$ has the property that with probability at least $\delta^c$ for a small constant $c >0$, $u=HDx$ has $m = \Theta(d/\ln(1/\delta))$ non-zero entries, each of magnitude $1/\sqrt{m}$. Conditioning on such a transformed vector $u=HDx$ puts a lot of structure on $u$, which simplifies the analysis of the product $Pu$. Indeed, if we consider a coordinate $(Pu)_i$, then this coordinate is $\Norm(0,\sum_j b_{i,j} u_j^2/q)$ distributed if we condition on the Bernoullis $b_{i,j}$. But $u_j^2$ is $1/m$ for precisely $m$ values of $j$ and $0$ for all others. Thus $\sum_j b_{i,j} u_j^2/q$ is distributed as $1/(qm)$ times a binomial distribution with $m$ trials and success probability $q$. One part of the analysis is thus to study this distribution. Secondly, if we consider $\|Pu\|^2$, then this is a linear combination of $k$ independent $\chi^2$ random variables, with the $i$'th being scaled by $\sum_j b_{i,j} u_j^2/q$. Hence we also need to understand the tail of such a distribution.

For the first step, i.e. showing $q = \Omega(\ln(1/\delta)/d)$, we argue that the sum of the coefficients $\sum_j b_{i,j} u_j^2/q$ deviates a lot from its expectation with reasonable probability. More precisely, notice that $\E[\sum_j b_{i,j} u_j^2/q] = (mq)/(mq) = 1$ and thus $\E[\sum_i \sum_j b_{i,j} u_j^2/q] = k$. But the sum of these coefficients is itself distributed as $1/(mq)$ times a binomial distribution with $mk$ trials and success probability $q$. The number of successes in such a binomial distribution deviates by additive $\Omega(\sqrt{\ln(1/\delta) (mkq)})$ from its expectation $mkq$ with probability at least $\delta^c$ for a small constant $c>0$. Intuitively, we need this deviation to be less than $\eps mkq$ to preserve the norm of $x$ (and thus $u$) to within $(1\pm \eps)$. This implies $\sqrt{\ln(1/\delta)(mkq)} = O(\eps mkq) \Rightarrow q=\Omega(\ln(1/\delta)/(\eps^2 mk)) = \Omega(1/m) = \Omega(\ln(1/\delta)/d)$.

In the second step, we now use the fact that we know that $q$ is sufficiently large, such that coordinates $2,\ldots,k$ of $Pu$ are reasonably well concentrated around their mean. What establishes the second lower bound on $q$, namely $q=\Omega(\varepsilon\min\{1,\ln^2(1/\delta)/(d\ln (1/\eps))\})$, is the possibility that the first coordinate $(Pu)_1$ may be so large that it alone distorts the norm $\|k^{-1/2}Pu\|^2$. In more detail, we show that with good probability, we have $\sum_{i=2}^{k} k^{-1}(Pu)_i^2 \in (1 \pm \eps)(k-1)/k$, i.e. on the last $k-1$ coordinates, the embedding $k^{-1/2}PHDx$ preserves the norm of $x$ as it should (we work with $k^{-1}\|Pu\|^2$ instead of $k^{-1/2}\|Pu\|$ to simplify the analysis - and since the later is a weaker statement by $\sqrt{1\pm\eps} \subset (1\pm\eps)$ it suffices to work with $k^{-1}\|Pu\|^2$). In this case, we show that unless $q$ is large enough, the single coordinate $k^{-1}(Pu)_1$ contributes more than $\eps$ to $k^{-1}\|Pu\|^2$ with probability more than $\delta$.

We now give the details of the proof outlined above. We first show the existence of the vector $x$ for which $u=HDx$ often has $m = O(d/\ln(1/\delta))$ coordinates of magnitude $1/\sqrt{m}$.

\paragraph{Hard Instance.}
Let $\varepsilon,\delta>0$ and set $\clbiiiiii$ to be the integer such that  $\clbiiiiii\leq\log_2\left(\log_2(1/\sqrt{2\delta})\right)\leq \clbiiiiii+1$ and define \begin{equation}\label{phdud2}x_i=\begin{cases}
\sqrt{\frac{1}{2^l}} \ \text{ if } i\leq 2^\clbiiiiii \\
0 \ \text{ else. }
\end{cases}
\end{equation}

We now notice that $Dx=x$ with probability $2^{-2^\clbiiiiii}\geq \sqrt{2\delta}$. Since the unnormalized Hadamard matrix is given recursively by 
\begin{gather*}
H_{2^{i}}=\left[\begin{array}{cc}
H_{2^{i-1}} & H_{2^{i-1}} \\
H_{2^{i-1}} & -H_{2^{i-1}}
\end{array}\right]
=\left[\begin{array}{cc}
H_{2^{l}} & \cdots\\
\vdots & \ddots\\
H_{2^{l}} & \cdots
\end{array}\right],
\end{gather*}
for $i\in\mathbb{N}$ and $x$ has $1$'s in the first $2^l$ places and zeros in the rest, we get $Hx=[H_{2^l}\mathbf{1},\ldots,H_{2^l}\mathbf{1}]^T/\sqrt{d}$, with $\mathbf{1}$ being the all-ones vector in $\mathbb{R}^{2^l}$. Now since we further know that for any $i$, the rows of the unnormalized Hadamard matrix are orthogonal, and that the first row of the unnormalized Hadamard matrix is all-ones, it follows that
$$
(Hx)_i=\begin{cases}
 \sqrt{\frac{2^\clbiiiiii}{ d}}\ \text{ if } i \equiv 0 \mod( 2^\clbiiiiii ) \\
0 \ \text{ else. }
\end{cases}
$$
Thus we conclude that $u:=Hx$ has $d/2^\clbiiiiii$ non-zero entries, all of value $\sqrt{2^\clbiiiiii/d}$. This is the vector $u$ we will analyze throughout the remainder of the lower bound proof.

Using the definition of $u$ we have that 
\begin{gather*}
\norm{Pu}\stackrel{d}{=}\sum_{i=1}^\clbiiiii \left(\sum_{j=1}^{\frac{d}{2^\clbiiiiii}} \sqrt{\frac{2^\clbiiiiii}{ dq}} b_{i,j} N_{i,j}\right)^2=\frac{2^{l}}{ dq}\sum_{i=1}^\clbiiiii \left(\sum_{j=1}^{\frac{d}{2^\clbiiiiii}}  b_{i,j} N_{i,j}\right)^2,
\end{gather*}
where the $b_{i,j}$'s are Bernoulli random variables with success probability $q$ and the $N_{i,j}$'s are $\Norm(0,1)$ distributed, all independent of each other. Conditioned on the outcome of the $b_{i,j}$'s it follows from linear combinations of independent normal random variables that 
\begin{gather*}
\sum_{i=1}^\clbiiiii \left(\sum_{j=1}^{\frac{d}{2^\clbiiiiii}}  b_{i,j} N_{i,j}\right)^2\stackrel{d}{=}\sum_{i=1}^\clbiiiii \left( \left(\sqrt{\sum_{j=1}^{\frac{d}{2^\clbiiiiii}}  b_{i,j}}\right) N_{i}\right)^2\stackrel{}{=}\sum_{i=1}^\clbiiiii b_i  N_{i}^2,
\end{gather*}
where the $b_i$'s are $\sum_{j=1}^{d/2^\clbiiiiii}  b_{i,j}$ and the $N_i$'s are independent standard normal random variables. Hence we conclude that the above is a weighted sum of $\chi^2$-variables.

What remains is the two steps described earlier where we analyze this distribution to derive the lower bound. We give the steps in the following two sections.

\subsection{First step $q=\Omega\left(\frac{\ln \frac{1}
{\delta}}{d}\right)$}\label{firststep}
As described in the proof sketch we need lower bounds on the tail probabilities for weighted sums of independent $\chi^2$-distributions, thus we now restate Theorem 7 from \cite{zhang2020nonasymptotic} in a slightly weaker form.
\begin{restatable}[\cite{zhang2020nonasymptotic}]{Lemma1}{sensseven}
\label{reformulation restatement 1} 
Let $g_1,\ldots,g_d$ be independent $N(0,1)$ random variables and $u_1,\ldots,u_d$ be non-negative numbers, then for constants $0<c_3$ and $C_3\geq 1$ we have that

$$
\begin{gathered}
c_3 \exp \left(-C_3 x^{2} /\|u\|_{2}^{2}\right) \leq \mathbb{P}\left[\sum_{i=1}^d u_i(g_i^2-1) \geq x\right] , \quad \forall 0 \leq x. 
\end{gathered}
$$
\end{restatable}
We will also need the following reverse Chernoff bound from \cite{reversechernoff} which we restate in a multiplicative version instead of an additive:
\begin{restatable}[\cite{reversechernoff}]{Lemma1}{senseight2}
\label{lem:revchernoff}
Let $X$ be binomial distributed with $r$ trials and success probability  $q \leq 1 / 4$. Then for any $0 \leq \alpha q \leq 1/4$ it holds that:
$$
\operatorname{Pr}\left[X \geq (1+\alpha) qr\right] \geq \frac{1}{4} \exp \left(-2 \alpha^2qr\right)
$$
\end{restatable}

With the above lemma stated we now present the first step in the proof of \Cref{phdtheorem2main}.
\paragraph{First step in proof of \Cref{phdtheorem2main}.}
\begin{proof}
We condition on the randomness in $HD$ resulting in the fixed vector $u$ as argued earlier. In this case, we start by showing that $\sum_i b_i$ is large with reasonable probability. Observe that $\sum_i b_i$ is binomial distributed with $r=kd/2^l$ trials and success probability $q$. Hence for $\alpha = \sqrt{\ln(1/(4^{4}\delta))/(8qr)}$, it follows from \Cref{lem:revchernoff} that either $\alpha q  > 1/4$ or $q > 1/4$ or $\Pr[\sum_i b_i \geq qr + \sqrt{\ln(1/(4^{4}\delta))qr/8}] \geq \delta^{1/4}$. 

If $q\geq 1/4$ we are done. Likewise, if $\alpha q  \geq 1/4$ then $q\geq1/(4\alpha)$ implying that $q \geq \sqrt{qr/(2\ln(1/(4^{4}\delta)))} \geq \Omega(\eps^{-2})$ by assumptions on $r=kd/2^l$, $k=\log(1/\delta)/\eps^2$ and $d/2^l\geq1$ and we are done again.

Thus what remains is the case $\Pr[\sum_i b_i \geq qr + \sqrt{\ln(1/(4^{4}\delta))qr/8}] \geq \delta^{1/4}$. Let us condition on $\sum_i b_i \geq qr + \sqrt{\ln(1/(4^{4}\delta))qr/8}$. Then by \Cref{reformulation restatement 1} with $x=0$ we get $\Pr[\sum_i b_i(N_i^2-1) \geq 0] \geq c_3$. This implies $\sum_i b_i N_i^2 \geq \sum_i b_i \geq qr + \sqrt{\ln(1/(4^{4}\delta))qr/8}$ with probability at least  $c_3 \delta^{1/4}$. But $ (2^l/(dq))(qr +  \sqrt{\ln(1/(4^{4}\delta))qr/8}) = k + \sqrt{\ln(1/(4^{4}\delta)) 2^{2 l} r/(8d^2q)} = k + \Omega(\sqrt{\ln(1/\delta) 2^l k/(qd)})$. Thus with probability at least $c_3 \delta^{1/4}$, we have $(2^l/(dq))\sum_i b_i N_i^2 \geq k + \Omega(\sqrt{\ln(1/\delta) 2^l k/(qd)})$. And since $\|Pu\|^2\stackrel{d}{=}(2^l/(dq))\sum_i b_i N_i^2 $ we also have that $\|Pu\|^2 \geq k + \Omega(\sqrt{\ln(1/\delta) 2^l k/(qd)})$ with probability  $c_3 \delta^{1/4}$. Further since we noticed (below \cref{phdud2}) that the probability of $HDx=u$ is at least $\sqrt{2\delta}$ it now follows what with probability at least $c_3 \delta^{3/4}$ we have that $$\frac{1}{k}\norm{PHDx} >1 + \Omega(\sqrt{\ln(1/\delta) 2^l /(kqd)}).$$ Thus for $\delta \leq c_3^{4}$ it follows that we must have 
\[
\Omega\left(\sqrt{\ln(1/\delta) 2^l /(kqd)}\right) \leq \eps 
\]
for $\frac{1}{k}\norm{PHDx}$ to satisfy \cref{eq:jl} (being a length preserving projection) with probability $\delta$,  which implies $q \geq\Omega(  \ln(1/\delta) 2^l/(\eps^{2}kd)) = \Omega(\ln(1/\delta)/d)$ where we have used that $2^l$ is $\Theta(\ln(1/\delta))$ by the choose of $l$, which completes the proof of the first step.
\end{proof}

\subsection{Second step $q=\Omega\left(\varepsilon\min\left\{1,\ln^2\left(1/\delta\right)/\left(d\ln \left(1/\eps\right)\right)\right\}\right)$}\label{secondsteplowerbound}
In this section we show the second step of the lower bound. Recall from the proof sketch that we use the result from the first step, giving $q = \Omega(\ln(1/\delta)/d)$. The basic idea is to show that there is a reasonably large probability that the first coordinate $(Pu)_1$ is so large that it distorts the embedding of $x$ by too much, even when all other coordinates behave well.

We now make some preliminaries and present some lemmas we will need in the proof of the second step.
By the the first step, we already have our claimed lower bound in \Cref{phdtheorem2main} whenever 
\[
\Theta\left(\max\left\{\ln(1/\delta)/d,\varepsilon\min\left \{1,\ln^2(1/\delta) /(d\ln(1/\eps))\right \}\right\}\right )=\Theta\left(\ln(1/\delta)/d\right),
\] 
so we now consider the cases where $\varepsilon,\delta,d$ are such that $$\Theta\left(\max\left\{\ln(1/\delta)/d,\varepsilon\min\left\{1,\ln^2(1/\delta) /(d\ln(1/\eps))\right\}\right\}\right)=\Theta\left(\varepsilon\min\left\{1,\ln^2(1/\delta) /(d\ln(1/\eps))\right\}\right),$$ and then show that for 
\begin{gather}\label{constansc6} \clbxii\ln(1/\delta)/d\leq q\leq \clbiii\varepsilon\min\left\{1,\ln^2\left(1/\delta\right) /\left(d\ln\left(1/\eps\right)\right)\right\},\end{gather}where $\clbxii$ is the constant from the lower bound $q \geq \clbxii \ln(1/\delta)/d$ and $\clbiii$ is a constant to be fixed later (but will be chosen less than 1), we have that the projection fails with at least $\delta$ probability. 

We construct our hard instance as in step one, except we adjust $l$ a bit (to deal with constants). We thus set $\clbiiiiii$ to be the integer such that  $\clbiiiiii\leq\log_2\left(\log_2(\left(1/\delta\right)^{\min\left\{1/50,\clbxii/\log_2(e)\right\}})\right)\leq \clbiiiiii+1$ and define \begin{equation*}x_i:=\begin{cases}
\frac{1}{\sqrt{2^l}} \ \text{ if } i\leq 2^\clbiiiiii \\
0 \ \text{ else. } 
\end{cases}
\end{equation*} 
It thus follows that  with probability $2^{-2^\clbiiiiii}\geq{\delta}^{\min\left\{1/50,\clbxii/\log_2(e)\right\}}$, the first $2^l$ signs in $D$ are $1$ , thus $Dx=x$ with at least probability ${\delta}^{\min\left\{1/50,\clbxii/\log_2(e)\right\}}$. We further notice that for the above $x$ we have that
$$
u_i:=(Hx)_i=\begin{cases}
 \sqrt{\frac{2^\clbiiiiii}{ d}}\ \text{ if } i \equiv 0 \mod( 2^\clbiiiiii ) \\
0 \ \text{ else. }
\end{cases}
$$
 We notice that the $u$ has $d/2^\clbiiiiii$ entries of size $\sqrt{2^\clbiiiiii/d}$ and 0 else, we let $m$ denote the number of non-zero entries. 
 
 We further notice if $\ln(1/\delta)/(qm)\leq \clbxiii$ then by the choose of $l$ and $m=d/2^l$ we have that $q$ is greater than $\ln^2 (1/\delta)\clbiiii/( \clbxiii d)$  and since  $\Theta\left(\max\left\{\ln(1/\delta)/d,\varepsilon\min\{1,\ln^2(1/\delta) /(d\ln(1/\eps))\}\right\}\right)=O(\ln^2(1/\delta)/d)$  we are done. Hence we may assume in the following that \begin{gather}\label{constantc7}
\ln(1/\delta)/(qm)\geq \clbxiii,
\end{gather} where $\clbxiii$ is at least $8$, and will be chosen larger later. 
 
Let now for $i=1,\ldots,k$, $Z_i$ denote a normalized sum of $m$ independent Bernoulli random variables $Z_i=(1/m)\sum_{j=1}^m b_{i,j}$ and $N_i$ denote a standard normal random variable, where all the $Z_i$'s and the $N_i$'s are independent of each other. Then for the $u$ described above, we have by linear combinations of independent normal distributions that:
\begin{gather*}
\norm{Pu}\stackrel{d}{=}
\sum_{i=1}^\clbiiiii \left(\sum_{j=1}^{\frac{d}{2^\clbiiiiii}} \sqrt{\frac{2^\clbiiiiii}{ dq}} b_{i,j} N_{i,j}\right)^2
\stackrel{d}{=}\sum_{i=1}^\clbiiiii \frac{1}{ q} Z_i N_{i}^2,
\end{gather*}
where we in the following will work with the later variable.

With the above preliminaries we now present the lemmas, used in the second step in the proof of \Cref{phdtheorem2main}. The proof of the lemmas can be found in \Cref{lowerboundappendix}.

In the spirit of the proof sketch of the second step, we now present \Cref{lowerbound1/qZN2}, which state that with good probability, the first coordinate of our projection vector, $Z_1 N_1^2/q$, is large.

\begin{restatable}{Lemma1}{senseight}
\label{lowerbound1/qZN2}
For $0<\varepsilon,\delta\leq 1/4$, $\clbiii$ sufficiently small (\cref{constansc6}), and $\clbxiii$ sufficiently large (\cref{constantc7}) we have with probability at least $\delta^{1/50+1/2+1/\pi}$ that
\begin{gather*}
\frac{1}{q}Z_1N_1^2\geq\frac{5\ln(1/\delta)}{\varepsilon}.
\end{gather*}
\end{restatable}

As noted in the proof sketch, we also want to show that the sum of the coordinates except $Z_1N_1^2/q$ have a good concentration around its mean: 
\begin{restatable}{Lemma1}{sensnine}
\label{lowerboundsumN}
For $0<\eps\leq1/4$ and $0<\delta\leq1/8$ we have with probability at least $\delta^{1/8}$ that 
\begin{gather*}
\sum_{i=2}^k \frac{1}{q}Z_iN_i^2\geq (1-3\varepsilon)(k-1).
\end{gather*}
\end{restatable}

We are now ready to put the above lemmas together and complete the proof of \Cref{phdtheorem2main}. 
\paragraph{Second step in proof of \Cref{phdtheorem2main}.}
\begin{proof} 
Let $0<\varepsilon\leq 1/4$ and $0<\delta\leq 1/8$. We now choose $\clbiii$ and $\clbxiii$ accordingly to \Cref{lowerbound1/qZN2}, thus we have with probability at least $\delta^{1/50+1/2+1/\pi}$ that $Z_1N_1^2/q\geq 5\ln(1/\delta)\varepsilon^{-1}$.
By \Cref{lowerboundsumN} we have $\sum_{i=2}^k Z_iN_i^2/q\geq (1-3\varepsilon)(k-1)$ with probability at least $\delta^{1/8}$. 

Thus we conclude by independence of the $Z_i$'s and the $N_i$'s that with probability $\delta^{1/50+1/2+1/\pi+1/8}$ we have for the vector $u$ that
\begin{gather*}
\|Pu\|^2\stackrel{d}{=} \sum_{i=1}^{k} \frac{1}{q}Z_iN_i^2\\
= \frac{1}{q}Z_1N_1^2+\sum_{i=2}^{k} \frac{1}{q}Z_iN_i^2\\
\geq 5\cphdii\ln(1/\delta)\varepsilon^{-1}+(1-3\varepsilon)(k-1)\\
= 5\varepsilon k+k-3\varepsilon k-1+3\varepsilon\\
=(1+\varepsilon)k+\varepsilon k-1+3\varepsilon\\
> (1+\varepsilon)k,
\end{gather*} 
where the last inequality follows by the assumptions on $\varepsilon\leq 1/4$ implying that $\varepsilon k=\cphdii \ln(1/\delta)\varepsilon^{-1}> 4\geq 1-3\varepsilon $. 

Now since we early notice that for the choose of $l$ ($\clbiiiiii\leq\log_2\left(\log_2(\left(1/\delta\right)^{\min\left\{1/50,\clbxii/\log_2(e)\right\}})\right)\leq \clbiiiiii+1$ ) we had $u=HDx$ happining with probability at least $\delta^{1/50}$ independently of the outcomes of the $b_{i,j}$'s and the $N_{i,j}$'s in $P$. Thus we conclude by the law of conditional probability that with probability at least $\delta^{1/50+1/2+1/\pi+1/8+1/50}\geq \delta$ we have $\norm{PHDx}>(1+\varepsilon)k$. Thus we have shown that for $\delta,\varepsilon$ less than sufficiently small constants, we must have $q\geq \clbiii\varepsilon\min\left \{1,\ln(1/\delta)/(d\ln (1/\eps))\right \}$ for the mapping $PHD$ to be a length preserving random projection with probability $1-\delta$.
\end{proof}

\section{Concentration Inequalities}\label{appendix}

\subsection{Inequalities for the Upper Bound}
\label{sec:upperin}
We now restate and present the proof of \Cref{ublemma1}
\begin{customlem}{\ref{ublemma1}}
\sensfour*
\end{customlem}

\begin{proof}
First notice that by a union bound and Markov's inequality we have that 
\begin{gather}\label{eqphdub1}
\ppp\left[ \max_{i=1,\ldots,k} Z_i>t\right] \leq k\ppp\left[Z_1>t\right] \leq k\E\left[\exp(cZ_1)\right]\exp(-ct),
\end{gather}
for $c>0$.

Now since that
\begin{gather*}
\E[\exp(cZ_1)]=\sum_{b' \in \{0,1\}^d}  \exp\left(\sum_{j=1}^d u_j^2 b'_{i,j}\right)\ppp\left[b'\right],
\end{gather*}
where $\sum_{j=1}^d u_j^2 b'_{i,j}$ is an convex function in $(u_1^2,\ldots,u_d^2)$, implying that $\exp(\sum_{j=1}^d u_j^2 b'_{i,j})$ is convex since it is the composition of the convex function $\sum_{j=1}^d u_j^2 b'_{i,j}$ and the increasing convex function $\exp(\cdot)$. Since a linear combination with positive scalars of convex functions is again a convex function, we conclude that $\E[\exp(cZ_1)]=\sum_{b' \in \{0,1\}^d}  \exp(\sum_{j=1}^d u_j^2 b'_{i,j})\ppp[b']$ is a convex function in $(u_1^2,\ldots,u_d^2)$. Now since we have that $(u_1^2,\ldots,u_d^2)$ lies in the set $\{x\in\mathbb{R}^d| x_i \in[0,1/m] \forall i\in 1,\ldots,d \text{ and } \sum_{i=1}^d x_i=1\}$ (which is a convex polytope), we must have that the function  $\E[\exp(cZ_1)]$ obtains its maximum on a vertex. The choice of vertex does not change the distribution of the random variable, so we can without loss of generality assume that $u_1^2,\ldots,u_m^2=1/m$ and $u_{m+1}^2,\ldots,u_{d}^2=0$. 

Using that the maximum of $\E[\exp(cZ_1)]$ is attained in such a vertex, we obtain that
\begin{gather}\label{eqphdub2}
\E\left[\exp(cZ_1)\right]\leq \E\left[\exp(\frac{c}{m}\sum_{i=1}^m b_{1,i})\right]=\left(\exp\left(\frac{c}{m}\right)q+\left(1-q\right)\right)^m\\\leq \exp\left(m\left(\exp\left(\frac{c}{m}\right)q-q\right)\right)=\exp\left(mq\left(\exp\left(\frac{c}{m}\right)-1\right)\right)\nonumber,
\end{gather}
where the first equality follows from the bernoulli trailes $b_{1,i}$ being independent and identically distributed. The second inequality uses that $0\leq(1+x)\leq \exp(x)$ for $x\in\mathbb{R}^+$.
Now setting $c=m\ln(t/q)$ (for $t>q$) and using \cref{eqphdub1} and \cref{eqphdub2}
\begin{gather*}
\ppp\left[\max_{i=1,\ldots,k} Z_i>t\right]\leq k\E\left[\exp\left(cZ_1\right)\right]\exp\left(-ct\right)\leq k\exp\left(mq\left(\frac{t}{q}-1\right)-mt\ln\frac{t}{q}\right).
\end{gather*}
Now setting $t=q/(2\alpha)>q$ we get that
\begin{gather*}
\ppp\left[\max_{i=1,\ldots,k} Z_i>t\right] \leq k\exp\left(mq\left(\frac{1}{2\alpha}-1-\frac{1}{2\alpha}\ln\frac{1}{2\alpha}\right)\right) =\\ k\exp\left(\frac{mq}{2\alpha}\left(1-2\alpha-\ln\frac{1}{2\alpha}\right)\right)\leq k\exp\left(-\frac{mq\ln\left(1/\alpha\right)}{32\alpha}\right),
\end{gather*}
where we in the second inequality have used that $\alpha\leq1/4$ so $(1-2\alpha-\ln(1/(2\alpha))\leq -\ln(1/\alpha)/16$.
\end{proof}

Next we give the proof of \Cref{ublemma3}. For this, we need the following technical lemma about linear combinations of independent Bernoulli random variables.

\begin{restatable}{Lemma1}{sensten}
\label{ublemma2a}
  Let $Z = \sum_{j=1}^d u^2_j b_{j}$ where $b_{j}$ are independent Bernoulli random variables with success probability $q$ and $u^2_j$ are positive real numbers bounded by $1/m$ and summing to $1$. We then have for $t>q$:
  \[
    \Pr\left[Z > t\right] < \left(\frac{t}{eq}\right)^{-mt}.
  \]
\end{restatable}

\begin{proof}
The proof follows the proof steps in \Cref{ublemma1}. For any $c \geq 0$, we have
  \[
    \E\left[\exp\left(cZ\right)\right] \leq \exp\left(mq\left(\exp\left(\frac{c}{m}\right)-1\right)\right).
  \]
  Thus by Markov's, we have for $c>0$
  \[
    \Pr\left[Z > t\right] = \Pr\left[\exp\left(cZ\right) > \exp\left(ct\right)\right] \leq \exp\left(mq\left(\exp\left(\frac{c}{m}\right)-1\right)\right) \exp\left(-ct\right) \leq \exp\left(mq\exp\left(\frac{c}{m}\right)-ct\right).
  \]
  Setting $c = m \ln(t/q)$ gives
  \[
    \Pr\left[Z > t\right] < \exp\left(\frac{mqt}{q} - mt \ln\frac{t}{q}\right ) = \exp\left(mt-mt\ln\frac{t}{q}\right) = \exp\left(-mt \ln\frac{t}{eq}\right) = \left(\frac{t}{eq}\right)^{-mt}.
  \]
  \end{proof}

With \Cref{ublemma2a} in place we now restate and prove \Cref{ublemma3}. 
\begin{customlem}{\ref{ublemma3}}
\sensfive*
\end{customlem}

  \begin{proof}
For simplicity we assume in the following that $\log_2 k$ is an integer. 
    For $j=0,\dots,\log_2 (k)/2$, let $E_j$ denote the event that there are at least $2^j/(j+1)^2$ indices $i$ such that $Z_i^2 \geq t/(2^{j+3})$ and let $E'_j$ denote the event that there are at least $k/(2^j (j+1)^2)$ indices $i$ with $Z_i^2 \geq t 2^{j-3}/k$. We claim that if $\sum_{i=1}^k Z_i^2 > t$, then one of the events $E_j$ or $E_j'$ must occur for some $j$. Before we prove this, we briefly motivate why we need the two separate events $E_j$ and $E_j'$. If we had only defined the events $E_j$, but let $j$ range all the way to $\log_2 k$, then either the $j=0$ or $j=\log_2 k$ term would dominate. The issue with this, is that the $(j+1)^2$ term is sub-optimal (i.e. non-constant) for $j=\log_2 k$. One could simply try to remove the $1/(j+1)^2$ term, but this would not work as $\sum_j 2^j \cdot t/2^{j+3}$ is $\omega(t)$. Including $1/(j+1)^2$ is precisely used to guarantee that $\sum_j 2^j/(j+1)^2 \cdot t/2^{j+3} = O(t)$.  For that reason, we define the events $E_j'$ that will handle the case of many indices with small values. 
    
    To prove that at least one event must occur, assume for the sake of contradiction that none of the events occur. Then:
    \begin{eqnarray*}
      \sum_{i=1}^k Z_i^2 &\leq& \\
      \sum_{i=1}^k \sum_{j=0}^\infty 1_{\{Z_i^2 \geq \frac{t}{2^{j+3}}\}} \frac{t}{2^{j+3}} &=& \\
      \sum_{j=0}^\infty \frac{t}{2^{j+3}} \sum_{i=1}^k 1_{\{Z_i^2 \geq \frac{t}{2^{j+3}}\}} &=& \\
      \sum_{j=0}^{\log_2 k } \frac{t}{2^{j+3}} \sum_{i=1}^k 1_{\{Z_i^2 \geq \frac{t}{2^{j+3}}\}} + \sum_{j=\log_2 k + 1}^{\infty} \frac{t}{2^{j+3}} \sum_{i=1}^k 1_{\{Z_i^2 \geq \frac{t}{2^{j+3}}\}} &\leq& \\
      \sum_{j=0}^{\log_2 (k)/2} \frac{t}{2^{j+3}} \sum_{i=1}^k 1_{\{Z_i^2 \geq \frac{t}{2^{j+3}}\}} + \sum_{j=0}^{\log_2 (k)/2} \frac{t}{2^{\log_2k - j+3}} \sum_{i=1}^k 1_{\{Z_i^2 \geq t/2^{\log_2 k - j+3}\}}  + \sum_{j=\log_2 k + 1}^{\infty} \frac{tk}{2^{j+3}}&\leq& \\
     \sum_{j=0}^{\log_2 (k)/2}  \frac{t2^j}{2^{j+3}(j+1)^2})+\sum_{j=0}^{\log_2 (k)/2} \frac{t2^{j-3}}{k} \sum_{i=1}^k 1_{\{Z_i^2 \geq \frac{t2^{j-3}}{k}\}}  + \frac{t}{8}&\leq& \\
      \frac{t}{8} \sum_{j=0}^{\log_2 (k)/2} \frac{1}{(j+1)^2}+\sum_{j=0}^{\log_2 (k)/2} \frac{kt2^{j-3}}{k(2^j (j+1)^2)}+ \frac{t}{8} &\leq& \\
      \frac{t}{4} \sum_{j=0}^{\infty}\frac{1}{(j+1)^2}+ \frac{t}{8} &=& \\
      \frac{t\pi^2}{4\cdot6} + \frac{t}{8} &<& t.
    \end{eqnarray*}
    We thus have $\Pr[\sum_{i=1}^k Z_i^2 > t] \leq \sum_{j=0}^{\log_2 (k)/2} \Pr[E_j] + \Pr[E'_j]$. To bound $\Pr[E_j]$, let $S$ be any subset of $2^j/(j+1)^2$ indices in $[k]$ and define the event $E_{j,S}$ which happens when all $i \in S$ satisfy $Z_i^2 \geq t/(2^{j+3})$. Notice since $t \geq 64\cdot24 e^3 q^2 k$ and $j \leq \log_2 (k)/2$ we have $t/2^{j+3}\geq 64\cdot24 e^3 q^2 k/(8k^{1/2})\geq 64\cdot3 e^3 q^2 k^{1/2}$ implying that the ratio of $\sqrt{t/2^{j+3}}$  with $q$ is larger than 1, \Cref{ublemma2a} is applicable with $Z\geq \sqrt{t/2^{j+3}}$. Now using an union bound over the events $E_{j,S}$ for any such set $S$, and that the $Z_i$'s on such sets are independent and identically distributed, combined with  \Cref{ublemma2a} yields that,
    \[
      \Pr\left[E_j\right] \leq \sum_S \Pr\left[E_{j,S}\right] \leq \binom{k}{2^j/(j+1)^2} \left(\sqrt{t/2^{j+3}}/\left(eq\right)\right)^{-m \sqrt{t/2^{j+3}} 2^j/\left(j+1\right)^2},
    \]
    and bounding $\binom{k}{2^j/(j+1)^2}$ by $k^{2^j/(j+1)^2}$, we obtain
    \[
      \Pr\left[E_j\right] \leq \exp\left(-\frac{2^j\left(m \sqrt{t/2^{j+3}} \ln\left(\sqrt{t/2^{j+3}}/\left(eq\right)\right) - \ln k\right)}{\left(j+1\right)^2}\right).
    \]
    For $t \geq 8e^2kq^2$ and $j \leq \log_2 (k)/2$, we have $\sqrt{t/2^{j+3}}/(eq) \geq \sqrt{8e^2kq^2/(8 \sqrt{k} e^2 q^2)}) \geq k^{1/4}$ and thus it follows that $\ln(\sqrt{t/2^{j+3}}/(eq)) \geq  \ln (k)/4$. Using $q \geq 8/(em)$ we also have $m \sqrt{t/2^{j+3}} \geq m \sqrt{8 e^2 k q^2/(8 \sqrt{k})} \geq m e q k^{1/4} \geq 8$. By this we then obtain
 \[
      \left(m \sqrt{t/2^{j+3}} \ln\left(\sqrt{t/2^{j+3}}/\left(eq\right)\right) - \ln k\right) \geq  m \sqrt{t/2^{j+3}} \ln\left(\sqrt{t/2^{j+3}}/\left(eq\right)\right)/2. 
    \]
    Thus letting $f(j)=2^{\frac{1}{2}j-5/2}m \sqrt{t}\ln(\sqrt{t/2^{j+3}}/(eq))/(j+1)^2$ we get that
\begin{gather*}
      \Pr\left[E_j\right] \leq \exp\left(-\left(\frac{2^{j-1}m \sqrt{t/2^{j+3}} \ln(\sqrt{t/2^{j+3}}/(eq))}{(j+1)^2}\right) \right)\\=\exp\left(-\frac{2^{\frac{1}{2}j-5/2}m \sqrt{t}\ln\left(\sqrt{t/2^{j+3}}/\left(eq\right)\right)}{\left(j+1\right)^2}\right)   
    =\exp\left(-f\left(j\right)\right).    
\end{gather*}

 Now using that $\ln(\sqrt{t/2^{j+3}}/(eq))\geq\ln(\sqrt{64\cdot 24e^3q^2k/(8\sqrt{k})}/(eq))\geq \ln\left (32\cdot 3 e\right )/2=\ln\left (96 e\right )/2$ for any $j\in0,\ldots,\log_2(k)/2$ and $t\geq64\cdot24e^3q^2k$ we get that the ratio between $f(j)$ and $f(j+1)$ for $j\in 0,\ldots,\log_2(k)/2-1$  is lower bounded by
 \begin{gather*}
\frac{f\left(j+1\right)}{f\left(j\right)}=  \frac{2^{1/2}\left(1-\ln \left(\sqrt{2}\right)/\ln \left(\sqrt{t/2^{j+3}}/\left(eq\right)\right)\right)\left(j+1\right)^2}{\left(j+2\right)^2}\geq   \frac{2^{1/2}\left(1-\ln \left(2\right)/\ln\left (96e\right )\right)\left(j+1\right)^2}{\left(j+2\right)^2}.
 \end{gather*}
By iteratively applying the above inequality for the ratio of consecutive terms of $f$ we get that for $j'\in1,\ldots,\log_2(k)/2$  that
 \begin{gather*}
 f\left(j'\right)\geq \frac{\left(2^{1/2}\left(1-\ln \left(2\right)/\ln\left (96e\right )\right)\right)^{j'} f\left(0\right)}{\left(j'+1\right)^2}\geq \frac{j'f\left(0\right)}{200},
 \end{gather*}
 where we in the last inequality have used that $((1-2\ln \left(2\right)/\ln\left (96e\right ))2^{1/2})^{j'}/\left( j'+1 \right )^2\geq j'/200$ for $j'\geq 0$.
 
 Now using the above inequality for $f$ we get by a geometric series argument that,
\begin{gather*}
      \sum_{j=0}^{\log_2 (k)/2} \Pr\left[E_j\right]\leq \exp\left(-f\left(0\right)\right)+\sum_{j=1}^{\log_2 (k)/2} \exp\left(-\frac{jf\left(0\right)}{200}\right) \\
      \leq \exp\left(-f\left(0\right)\right)+\frac{\exp\left(-\frac{f\left(0\right)}{200}\right)}{1-\exp\left(-\frac{f\left(0\right)}{200}\right)}\leq 3\exp\left(-2^{-5/2} \cdot m \sqrt{t}\ln\left(\sqrt{t/2^{3}}/\left(eq\right)\right)/200\right),
\end{gather*}
where we in the last inequality have used that $f(0)=2^{-5/2} \cdot m \sqrt{t}\ln(\sqrt{t/2^{3}}/(eq))\geq 250$, to say that $1/(1-\exp(-f(0)/200))\leq 2$.

Next we bound $\Pr[E'_j]$ . Again by a union bound over all sets of $k/(2^j (j+1)^2)$ indices and \Cref{ublemma2a}, we get:
    \[
      \Pr[E'_j] \leq \binom{k}{k/\left(2^j (j+1)^2\right)} \left(\sqrt{t 2^{j-3}/k}/\left(eq\right)\right)^{-m \sqrt{t 2^{j-3}/k} \cdot k/\left(2^j \left(j+1\right)^2\right)}.
    \]
    Bounding
$
      \binom{k}{k/(2^j (j+1)^2)}$ from above by $ (e2^j (j+1)^2)^{k/(2^j (j+1)^2)}$
    we get that
    \[
      \Pr\left[E'_j\right] \leq \exp\left(- \frac{k}{2^j (j+1)^2}\cdot \left(m \sqrt{t 2^{j-3}/k} \ln\left(\sqrt{t 2^{j-3}/k}/\left(eq\right)\right) - \ln\left(e2^j \left(j+1\right)^2\right)\right)\right).
    \]
    For $t \geq 24 e^3 k q^2$, we have $\sqrt{t 2^{j-3}/k}/(eq) \geq \sqrt{ 3 e 2^j}$. Since $(j+1)^2 \leq 3 \cdot 2^j$ for all $j \geq 0$, $\sqrt{ 3 e 2^j}$ is at least $\sqrt{e 2^{j/2}(j+1)} \geq (e 2^j (j+1)^2)^{1/4}$ and thus $\ln(\sqrt{t 2^{j-3}/k}/(eq)) \geq  \ln(e2^j(j+1)^2)/4$. For $q \geq 8/(em)$, we also have $m \sqrt{t 2^{j-3}/k} \geq m \sqrt{3 e^3 q^2}\geq 8$ and hence:
    \[
      m \sqrt{t 2^{j-3}/k} \ln\left(\sqrt{t 2^{j-3}/k}/\left(eq\right)\right) - \ln\left(e2^j \left(j+1\right)^2\right) \geq  m \sqrt{t 2^{j-3}/k} \ln\left(\sqrt{t 2^{j-3}/k}/\left(eq\right)\right)/2.
    \]
    
Now let $g(j)=m\sqrt{tk} \ln(\sqrt{t 2^{j-3}/k}/(eq)/((j+1)^22^{1/2j+5/2})$ then we have
\begin{gather*}
      \Pr\left[E'_j\right] \leq \exp\left(-\frac{km \sqrt{t 2^{j-3}/k} \ln\left(\sqrt{t 2^{j-3}/k}/\left(eq\right)\right)}{\left(j+1\right)^22^{j+1}} \right) =\exp\left(-g\left(j\right)\right).
\end{gather*}

 Now for any $j\in0,\ldots,\log_2(k)/2$ and $t\geq64\cdot24e^3q^2k$ it holds that $\ln(\sqrt{t 2^{j-3}/k}/(eq))$ is at least  $\ln(\sqrt{32\cdot24e^3q^2k /(8k)}/(eq))\geq \ln  \left (192e \right )/2$. This implies that the ratio between $g(j+1)$ and $g(j)$ for $j\in 0,\ldots,\log_2(k)/2-1$ is
 \begin{gather*}
\frac{g\left(j+1\right)}{g\left(j\right)}= \frac{2^{-1/2}\left(1+\ln \left(\sqrt{2}\right)/\ln\left(\sqrt{t 2^{j-3}/k}/(eq)\right)\right)(j+1)^2}{(j+2)^2}\leq  \frac{2^{-1/2}\left(1+\ln \left(2\right)/\ln  \left (192e \right )\right)(j+1)^2}{(j+2)^2}. 
 \end{gather*}
 Now iteratively using the above relation on the ratio between $g(j+1)$ and $g(j)$ and that $g(\log_2(k)/2)=k^{1/4}m\sqrt{t}\ln \left(t/(8e^2q^2\sqrt{k})\right)/(2^{7/2}(\ln(k)/2+1)^2)$ we get for $j'\in0,\ldots,\log_2(k)/2-1$ that

\begin{gather}
 g(j')\geq 
\frac{ \left (\log_2\left(k\right)/2+1\right )^2 g\left(\log_2\left(k\right)/2\right)}{\left(2^{-1/2}\left(1+\ln \left(2\right)/\ln  \left (192e \right )\right)\right)^{\left(\log_2(k)/2-j'\right)}\left(j'+1\right )^2}  \nonumber\\
\geq \frac{k^{1/4}m\sqrt{t}\ln \left(t/(8e^2q^2\sqrt{k})\right)}{ \left(2^{-1/2}\left(1+\ln \left(2\right)/\ln  \left (192e \right )\right)\right)^{\left(\log_2(k)/2-j'\right)}22k^{1/8}2^{7/2}}\nonumber\\
 \geq \frac{k^{1/8}m\sqrt{t}\ln \left(t/(8e^2q^2\sqrt{k})\right)}{\left(2^{-1/2}\left(1+\ln \left(2\right)/\ln  \left (192e \right )\right)\right)^{\left(\log_2(k)/2-j'\right)}22\cdot2^{7/2}}\label{ubineq1}\\
 \geq \frac{\left(\log_2(k)/2-j'\right)k^{1/8}m\sqrt{t}\ln \left(t/(8e^2q^2\sqrt{k})\right)}{200\cdot22\cdot2^{7/2}},\nonumber
 \end{gather}
 where we in the second inequality have used that for $j'\geq 0$ we have $(j'+1)^2\leq 22\cdot 2^{j'/4}\leq 22\cdot k^{1/8}$  and where we in the last inequality have used that for $j'=0,\ldots,\log_2(k)/2-1$  we have 
 \[
 \left(2^{-1/2}\left(1+\ln (2)/(\ln (192e))\right)\right)^{-\left(\log_2(k)/2 -j'\right)}\geq (\log_2(k)/2-j')/200.
 \]
 Now using that \cref{ubineq1}, also holds for $j'=\log_2 (k)/2$, and a geometric series argument we get that,
\begin{gather*}
      \sum_{j=0}^{\log_2 (k)/2} \Pr\left[E'_j\right]\\
      \leq \exp\left(-\frac{k^{1/8}m\sqrt{t}\ln \left(t/\left(8e^2q^2\sqrt{k}\right)\right)}{22\cdot2^{7/2}}\right)+\sum_{j'=0}^{\log_2 (k)/2-1} \exp\left(-
\frac{\left(\log_2(k)/2-j'\right)k^{1/8}m\sqrt{t}\ln \left(t/\left(8e^2q^2\sqrt{k}\right)\right)}{200\cdot22\cdot2^{7/2}}\right)
\\
 \leq \exp\left(-\frac{k^{1/8}m\sqrt{t}\ln \left(t/\left(8e^2q^2\sqrt{k}\right)\right)}{22\cdot2^{7/2}}\right)+\frac{\exp\left(-k^{1/8}m\sqrt{t}\ln \left(t/\left(8e^2q^2\sqrt{k}\right)\right)/(200\cdot22\cdot2^{7/2})\right)}{1-\exp\left(-k^{1/8}m\sqrt{t}\ln \left(t/\left(8e^2q^2\sqrt{k}\right)\right)/(200\cdot22\cdot2^{7/2})\right)}\\
 \leq 11\exp\left(-\frac{k^{1/8}m\sqrt{t}\ln \left(t/\left(8e^2q^2\sqrt{k}\right)\right)}{200\cdot22\cdot2^{7/2}}\right),
\end{gather*}
where we in the last inequality have used that $k^{1/8} m \sqrt{t}\ln(t/(8e^2q^2\sqrt{k}))/(200\cdot22\cdot 2^{7/2}) \geq 1/10$ .

By the above upper bounds on $\sum_{j=0}^{\log_2 (k)/2} \Pr[E'_j]$ and $\sum_{j=0}^{\log_2 (k)/2} \Pr[E_j]$ we can conclude that 
\begin{gather*}
\Pr\left[\sum_{i=1}^k Z_{i}^2\geq t\right]\\
\leq 14\exp\left(-\min\left\{k^{1/8}m\sqrt{t}\ln \left(t/\left(8e^2q^2\sqrt{k}\right)\right)/\left(200\cdot22\cdot2^{7/2}\right), m \sqrt{t}\ln\left(\sqrt{t/2^{3}}/\left(eq\right)\right)/\left(200\cdot2^{5/2}\right)\right\}\right)\\
\leq14\exp\left(-m\sqrt{t}/\left(200\cdot2^{7/2}\right)\min\left\{k^{1/8}\ln \left(t/\left(8e^2q^2\sqrt{k}\right)\right )/22,\ln\left(t/\left(8e^2q^2\right)\right)\right\}\right)\\
\leq 14\exp\left(-\frac{m\sqrt{t}\ln\left(\sqrt{t/2^{3}}/\left(eq\right)\right)}{200\cdot44\cdot2^{5/2}} \right),
    \end{gather*}
where we have used that the second term in the $\min$ is always smallest, when it is scaled by $1/44$, this follows from the assumption about $t \geq 64\cdot24 e^3 k q^2$ implying that for any such given t there exist $\tilde{c}\geq1$ such that $t=\tilde{c}8e^2kq^2$ and we get that
the first term in the $\min$ is equal to $k^{1/8}\ln \left(\tilde{c}\sqrt{k}\right )/22=k^{1/8}(\ln \left(\tilde{c}\right )+\ln\left (k\right )/2)/22$ and the second term in the $\min$ is equal to $\ln \left(\tilde{c}\right)+\ln \left(k\right)$, where by the claim follows.

\end{proof}

We now restate and present the proof of \Cref{ublemma4}.

\begin{customlem}{\ref{ublemma4}}
\senssix*
\end{customlem}

\begin{proof}
In the following we assume for simplicity that $\log_2 (k)$ and $\log_2(t)$ are integers.
We proceed in a somewhat similar fashion as in the proof of \Cref{ublemma3}. For $j=\log_2t,\ldots, \log_2k$ let $E_j$ be the event that there are at least $2^{j-1}/(j-\log_2(t)+1)^2$ indices such that $Z_i^2\geq t/2^{j+1}$. Assume that none of the events $E_j$ occurs, we then have that
\begin{align*}
&\sum_{i=1}^k Z_i^2  \\
&\leq\sum_{i=1}^k \sum_{j=\log_2(t)}^\infty 1_{\{Z_i^2 \geq \frac{t}{2^{j+1}}\}} \frac{t}{2^{j+1}}  \\
&= \sum_{j=\log_2(t)}^\infty \frac{t}{2^{j+1}} \sum_{i=1}^{k}1_{\{Z_i^2 \geq \frac{t}{2^{j+1}}\}}  \\
&=\sum_{j=\log_2(t)}^{\log_2 k } \frac{t}{2^{j+1}} \sum_{i=1}^{k}1_{\{Z_i^2 \geq \frac{t}{2^{j+1}}\}} + \sum_{j=\log_2 k + 1}^{\infty} \frac{t}{2^{j+1}} \sum_{i=1}^{k}1_{\{Z_i^2 \geq \frac{t}{2^{j+1}}\}}  \\
&\leq \sum_{j=\log_2(t)}^{ \log_2 k} \frac{t2^{j-1}}{2^{j+1}\left(\log_2(t)-j+1\right)^2}  + \sum_{j=\log_2 k + 1}^{\infty} \frac{tk}{2^{j+1}} \\
&\leq \frac{t}{4}\sum_{j=1}^{\infty} \frac{1}{j^2} + \frac{t}{4}\sum_{j=0}^{\infty} \frac{1}{2^{j} }\\
&\leq \frac{t\pi^2}{24} +\frac{t}{2}<t,
\end{align*}
where the first inequality follows by $Z_i^2\leq 1$, so the sum of the terms $1_{\{Z_i^2 \geq t/2^{j+1}\}}t/2^{j+1}$ starting at $j=\log_2(t)$ is always greater than $Z_i^2$. Thus we conclude that one of the events $E_j$ happens when $ \sum_{i=1}^k Z_i^2\geq t$. Now by an union bound over the events $E_j$ we have 
\begin{gather*}
\ppp\left[\sum_{i=1}^k Z_i^2\geq t\right] \leq \sum_{j=\log_2(t)}^{\log_2(k)} \ppp\left[E_j\right] .
\end{gather*}
When $E_j$ happens we know that there is a set $S$ of $2^{j-1}/(j-\log_2(t)+1)^2$ indices such that for $i\in S$ we have $Z_i^2\geq t/2^{j+1}$. Thus the probability of each $E_j$ can be bounded by using a union bound over all such possible sets of indices ($k$ choose $2^{j-1}/(j-\log_2(t)+1)^2$). Now using that the $Z_i$'s are independent and identically distributed, the probability of each of the sets $S$ splits into a product of probabilities $\ppp\left[Z_i^2\geq t/2^{j+1}\right] $, where \Cref{ublemma2a} can be used to bound each of these probabilities. We note that \Cref{ublemma2a} with $Z\geq \sqrt{t/2^{j+1}}$ is applicable since $\sqrt{t/2^{j+1}}/q\geq \sqrt{\cubi\ln(n)/(2k)}/(\cubii\varepsilon)=\sqrt{\cubi/(2\cubii^3)}\geq e^4$, where we have used the assumption that $t\geq \cubi\ln(n)$. We now get that:
\begin{gather*}
\ppp\left[E_j\right] 
\leq \binom{k}{2^{j-1}/\left(j-\log_2(t)+1\right)^2}\left(\sqrt{t/2^{j+1}}/\left(eq\right)\right)^{\sqrt{t/2^{j+1}}m2^{j-1}/(j-\log_2(t)+1)^2}\\
\leq \exp\left(-\frac{2^{j-1}\left(\sqrt{t/2^{j+1}}m\ln \left(\sqrt{t/2^{j+1}}/(eq)\right)-\ln \left(ek(j-\log_2(t)+1)^2/2^{j-1}\right)\right)}{(j-\log_2(t)+1)^2}\right),
\end{gather*}
where the last inequality follows by $\binom{k}{2^{j-1}/(j-\log_2(t)+1)^2}\leq \left(ek(j-\log_2(t)+1)^2/2^{j-1}\right)^{2^{j-1}/(j-\log_2(t)+1)^2}$.

To evaluate the term $\sqrt{t/2^{j+1}}m\ln \left(\sqrt{t/2^{j+1}}/(eq)\right)-\ln \left(ek(j-\log_2(t)+1)^2/2^{j-1}\right)$ we notice the following four relations for $j=\log_2(t),\ldots,\log_2(k)$
\begin{gather*}
\sqrt{t/2^{j+1}}m\geq \sqrt{\cubi\ln\left(n\right)/\left(2k\right)}\cubiiii d/\ln\left(n\right) \geq\sqrt{\cubi\varepsilon^2/\left(2\cubii\right)}\cubiiii k/\ln\left(n\right) \geq \sqrt{\cubi\cubiii/2}\cubiiii \varepsilon^{-1}\geq e^4\varepsilon^{-1},
\end{gather*}
\begin{gather*}
 \left(\sqrt{t/2^{j+1}}/(eq)\right)\geq \sqrt{\cubi\ln\left(n\right)/\left(2k\right)}/\left(e\cubii\varepsilon\right)=\sqrt{\cubi/\left(2e^2\cubii^3\right)}\geq e^3,
 \end{gather*}
 \begin{gather*}
\frac{ek}{2^{j-1}}
 \leq  e2k/t \leq e2\cubiii/\left(\cubi\varepsilon^2\right)\leq 1/\left(e^7\varepsilon^2\right),
 \end{gather*}
 \begin{gather*}
 j-\log_2(t)+1\leq \log_2(k/t)+1\leq\log_2\left(\cubiii/\left(\cubi\varepsilon^2\right)\right)+1=\log_2\left(2\cubiii/\left(\cubi\varepsilon^2\right)\right)\leq \log_2\left(1/\left(e^8\varepsilon^2\right)\right),
 \end{gather*}
where we have used that $c_1\geq 1/\cubiiii$ $t\geq \cubi\ln(n)$, $k=\cubii\varepsilon^{-2}\ln(2)$ and $d\geq k$. By the above relations we conclude that for sufficiently small $\varepsilon$, we have that \begin{gather*}
\sqrt{t/2^{j+1}}m\ln \left(\sqrt{t/2^{j+1}}/(eq)\right)-\ln \left(ek(j-\log_2(t)+1)^2/2^{j-1}\right)\geq  \sqrt{t/2^{j+1}}m\ln \left(\sqrt{t/2^{j+1}}/(eq)\right)/2.
\end{gather*}
Hence for such $\varepsilon$ and $f(j)=2^{j/2-5/2}\sqrt{t}m\ln \left(\sqrt{t/2^{j+1}}/(eq)\right)/(j-\log_2(t)+1)^2$ we have that
\begin{gather*}
\ppp\left[E_j\right] 
\leq \exp\left(-\frac{2^{j-1}\sqrt{t/2^{j+1}}m\ln \left(\sqrt{t/2^{j+1}}/\left(eq\right)\right)/2}{\left(j-\log_2(t)+1\right)^2}\right)=\exp\left(-f\left(j\right)\right).
\end{gather*}

Now using the assumptions that $t\geq \cubi\ln(n)$ and $q=\cubiii\varepsilon$ we get that $\sqrt{t/2^{j+1}}/(eq)\geq \sqrt{\cubi/2\cubii^3}/e\geq e^3$ such that for $j=\lg_2t,\ldots,\lg_2(k)-1$
\begin{gather*}
 \frac{f\left(j+1\right)}{f\left(j\right)}\geq \frac{\left(j-\log_2\left(t\right)+1\right)^2 \left(1-\ln \left(2\right)/6\right)\sqrt{2}}{\left(j+1-\log_2\left(t\right)+1\right)^2},
 \end{gather*}
using this iteratively we get that for $j'\in 1,\ldots,\log_2(k)-\log_2(t)$ 

\begin{gather*}
 f(\log_2(t)+j')\geq \frac{\left  (\left (1-\ln \left(2\right)/6\right)\sqrt{2}\right )^{j'}f(\log_2t)}{(j'+1)^2}\geq \frac{j'f(\log_2t)}{150},
 \end{gather*}
 where the last inequality follows by $\left  (\left (1-\ln \left(2\right)/6\right )\sqrt{2}\right )^{j'}/(j'+1)^2 \geq j'/150$ for $j'>1$. 
 
Now using a geometric series argument we get that
\begin{align*}
&\ppp\left[\sum_{i=1}^k Z_i^2\geq t\right] \\
&\leq \sum_{j=\log_2(t)}^{\log_2(k)}   \ppp\left[E_{j}\right]\\
&\leq \sum_{j=\log_2(t)}^{\log_2(k)} \exp\left(-f(j)\right)\\
&\leq \exp\left(-f(\log_2t)/150\right)+\sum_{j=1}^{\infty} \exp\left(-jf(\log_2t)/150\right)\\
&\leq 2\frac{\exp(-f(\log_2t)/150)}{1-\exp(-f(\log_2t)/150)}\\
&\leq 2\frac{\exp\left (-tm\ln(1/(\sqrt{2}eq))/(600\sqrt{2})\right )}{1-\exp\left (-tm\ln(1/ (\sqrt{2}eq))/(600\sqrt{2})\right )}.
\end{align*}
Now using that $t\geq2c_1^3e^8\ln(n)$ and $\varepsilon\leq \cubiii^{-1}/(4e)$ so $\ln(1/(\sqrt{2}eq))\geq \ln(2)$ we end up with the following inequality $t\ln(1/ (\sqrt{2}eq))/(600\sqrt{2})\geq\cubii^3 e^8 \ln (2)/(300\sqrt{2})\ln(n)\geq 4\cubii^3$ and since $m\geq 1$ we conclude that

\begin{gather*}
\ppp\left[\sum_{i=1}^k Z_i^2\geq t\right] \leq 2\frac{\exp\left (-tm\ln(1/ (\sqrt{2}eq))/(600\sqrt{2})\right )}{1-\exp\left (-tm\ln(1/ (\sqrt{2}eq))/(600\sqrt{2})\right )}
\leq2\frac{n^{-4\cubiiiii^3}}{1-n^{-4\cubiiiii^3}}
\leq 3n^{-4\cubiiiii},
\end{gather*}
where we in the last inequality have assumed that $n\geq 2$ and used that $\cubiii\geq1$, which completes the proof.
 
\end{proof}

\subsection{Inequalities for the Lower Bound}\label{lowerboundappendix}
In this section we proof \Cref{lowerbound1/qZN2} and \Cref{lowerboundsumN}. \Cref{lowerbound1/qZN2} states that the first coordinate $Z_1N_1^2/q$ is $\Omega(\varepsilon k)$ with good probability and \Cref{lowerboundsumN} says that $\sum_{i=2}^k Z_iN_i^2/q$ is $\Omega(k)$ with good probability, which we combined in (\Cref{secondsteplowerbound})  (the second step in the lower bound proof) to say that the sum of them became to large. To show \Cref{lowerbound1/qZN2} and \Cref{lowerboundsumN} we first recall the preliminaries for the second step of the lower bound (\Cref{secondsteplowerbound}). After the preliminaries we proof \Cref{lowerbound1/qZN2} via 4 helping lemmas and lastly we proof \Cref{lowerboundsumN}. Recall from \Cref{secondsteplowerbound}:

We consider the cases where $\varepsilon,\delta,d$ are such that 
\begin{equation}\label{constansc6a}
\clbxii\ln(1/\delta)/d\leq q\leq \clbiii\varepsilon\min\{1,\ln^2(1/\delta) /(d\ln(1/\eps)\}
\end{equation} where $\clbxii$ is the constant from \Cref{phdtheorem2main} and $\clbiii$ is a constant to be fixed later and will be chosen less than $1$.

We have $m=d/2^l$ where $\clbiiiiii\leq\log_2\left(\log_2(\left(1/\delta\right)^{\min\{1/50,\clbxii/\log_2(e)\}})\right)\leq \clbiiiiii+1$ implying that  $$ m\leq2d/(\min\{1/50,\clbxii/\log_2(e)\}\lg_2 (1/\delta))\leq2d/(\min\{1/50,\clbxii/\log_2(e)\}\ln(1/\delta)),$$ and  $$m\geq d/(\min\{1/50,\clbxii/\log_2(e)\}\lg_2 (1/\delta))\geq d/(\min\{1/50,\clbxii/\log_2(e)\} \log_2 (e)\ln(1/\delta)).$$ 

We notice that for $q$'s as in \cref{constansc6a} and the above $m$ we have that
\begin{equation}\label{boundonration}
    \min\{1/50,\clbxii/\log_2(e)\}\log_2(e) \ln(1/\delta)/\clbii\geq \ln(1/\delta)/qm \geq \clbiiii\ln(1/\eps) /(2\clbiii\varepsilon),
\end{equation} especially that $1/(qm) \leq1$.

We have that \begin{gather}\label{constantc7a}
\ln(1/\delta)/(qm)\geq \clbxiii,
\end{gather} 
where $\clbxiii$ is at least $8$, and will be chosen larger later. 

We consider the random variables $Z_1N_1^2/q$ and $\sum_{i=2}^k Z_i N_{i}^2/q$, where the $Z_i$'s denotes normalized sums of independent Bernoulli random variables $Z_i=(1/m)\sum_{j=1}^m b_j$ and the $N_i$'s denotes standard normal random variable, where all the $Z_i$'s and the $N_i$'s are independent of each other.  

We now present at technical lemma that we will need in the following proofs.

\begin{restatable}{Lemma1}{senseleven}
\label{lwbi}
For $a,x\in\mathbb{R}$ such that $0\leq x\leq 1$ and $0\leq ax\leq 1$ we have that
\begin{gather*}
\left(1-x\right)^a\leq \left(1-ax/2\right).
\end{gather*}
\end{restatable}

\begin{proof}
Cases $x=0,1$ and $ax=0$ can be realised by insertion, and the case $ax=1$ corresponds to $(1-x)^{1/x}\leq 1/2$ which holds. Now for the remainding cases we first note by Taylor expansion of $\ln(1-x)=-\sum_{i=1}^\infty x^i/i$ that $(1-x)^a=\exp(-a\sum_{i=1}^\infty x^{i}/i)$ and $(1-ax/2)=\exp(-\sum_{i=1}^\infty (ax/2)^i/i)$. So it suffices to show that $\sum_{i=1}^\infty (ax/2)^i/i \leq a\sum_{i=1}^\infty x^{i}/i$. Now using that $ax\leq1$ and that a geometric series with common ratio of $1/2$ equals $2$ we get that $\sum_{i=1}^\infty (ax/2)^i/i=(ax/2)\sum_{i=1}^\infty \frac{(ax/2)^{i-1}}{i}\leq (ax/2)2=ax$. We also have that $ax\leq a\sum_{i=1}^\infty x^i/i$. Hence we conclude that $\sum_{i=1}^\infty (ax/2)^i/i \leq a\sum_{i=1}^\infty x^i/i $ which proofs the claim.
\end{proof}

We will now present and proof \Cref{lowerboundz1}, \Cref{remarklowerboundstep2} and \Cref{lowerboundz1i} which combined yield that with good probability we have a lower bound of $\Theta(\varepsilon^{-1})$ on the scaled binomial $Z_1/q$.

\begin{restatable}{Lemma1}{senstwelve}
\label{lowerboundz1}
Let $0<\varepsilon,\delta\leq 1/4$. Let further $\clbi\leq 1$ and $L=\clbi\ln(1/\delta)/\ln\left (\ln(1/\delta)/(qm)\right ) $ if $m/L\geq1$, $qm/L\leq 1$ and $\clbiii$ (\cref{constansc6a}) is chosen so small that $\clbiiii /(2\clbiii)$ is greater than $2$. We then have with probability at least $\delta^{\clbi}$ that:
\begin{gather*}
\frac{Z_1}{q}= \frac{1}{q}\sum_{i=1}^{m}\frac{1}{m} b_{1,i}
\geq \frac{c_8\clbi}{\varepsilon\sqrt{\clbiii}},
\end{gather*}
with $c_8=\ln(2) \sqrt{\clbiiii }/(4\sqrt{2}).$
\end{restatable}
\begin{proof}
The idea of the proof is to divide the $m$ Bernoulli trails inside the sum $Z_1=\sum_{i=1}^m \frac{1}{m}b_{1,i}$ into $L$ disjoint buckets of size $m/L$(we choose $\clbi$ such that the bucket size is an integer), and then calculate the probability that all the buckets have at least one success, and here by get the above lower bound on $Z_1/q$.

Using that the buckets are disjoint so the events of buckets having a success in it is independent of each other the probability of having at least one success in every disjoint bucket is $(1-(1-q)^{m/L})^L$. Now using \Cref{lwbi} with $x=q$ and $a=m/L$ we get that $\left (1-(1-q)^{m/L}\right )^L\geq \left (1-\left (1-(qm)/(2L)\right )\right )^L=\left ((qm)/(2L)\right )^L$. Now plugging $L$ into this expression we get that 
\begin{gather*}
\left (\frac{qm}{2L}\right )^L=
\left (\frac{\ln \left (\ln\left(1/\delta\right)/\left(qm\right)\right )qm}{2\clbi\ln (1/\delta) }\right )^{\clbi\ln\left(1/\delta\right)/\ln\left (\ln\left(1/\delta\right)/\left(qm\right)\right ) }\\
= \left (\frac{\ln \left (\ln\left(1/\delta\right)/\left(qm\right)\right )}{2\clbi}\right )^{\clbi\ln\left(1/\delta\right)/\ln\left (\ln\left(1/\delta\right)/\left(qm\right)\right ) }  \delta^{\clbi}
\geq \delta^{\clbi},
\end{gather*}
where the last inequality follows from the assumption that $\ln (1/\delta)/(qm)\geq 8$ (\cref{constantc7a})  so the first term in the second to last expression is lower bounded by 1. Hence with probability at least $\delta^{\clbi}$ we have that all the disjoint $L$ buckets have at least one success and hence on this event $Z_1/q\geq L/(qm)$. Plugging L into the expression, using that $x/\ln x$ is increasing for $x\geq3$ and that $\ln (1/\delta)/(qm)$ is lower bounded by  $ \clbiiii\ln(1/\eps) /(2\clbiii\varepsilon)$ (\cref{boundonration}) which is at least 3 by assumptions on $\clbiii$ and $\varepsilon\leq 1/4$, it follows that

\begin{gather}\label{eqsenstwelve}
\frac{1}{q}Z_1
\geq \frac{\clbi\ln\left(1/\delta\right)}{qm \ln\left (\ln\left(1/\delta\right)/\left(qm\right)\right )}
\geq  \frac{\clbi\clbiiii\ln(1/\eps) }{2\clbiii\varepsilon\ln \left (\clbiiii\ln(1/\eps) /(2\clbiii\varepsilon)\right )}.
\end{gather}
Since $ \clbiiii/(2\clbiii)\geq2$ by assumption it holds that 
$\ln \left (\clbiiii\ln(1/\eps) /(2\clbiii\varepsilon)\right )$ is less than or equal to $\ln \left ((\clbiiii /(2\clbiii\varepsilon))^2\right )$, thus

\begin{gather*}
\frac{\ln(1/\eps) }{\ln \left (\clbiiii\ln(1/\eps) /(2\clbiii\varepsilon)\right )}\geq \frac{\ln(1/\eps) }{2\ln \left (\clbiiii /(2\clbiii\varepsilon)\right )}.
\end{gather*}
Now using that $x/(x+a)$ with $a,x>0$ is increasing in $x$, with $a=\ln\left(4/(\clbiii\clbiiii)\right)$, $x=\ln (1/\varepsilon)$ and $\ln (1/\varepsilon)\geq \ln 2$ it follows that 
\begin{gather*}
\frac{\ln(2) }{2(\ln \left (\clbiiii /(2\clbiii)\right )+\ln(2))}\geq\frac{\ln(2) }{4\ln \left (\clbiiii /(2\clbiii)\right )}.
\end{gather*} 
Plugging this into \cref{eqsenstwelve} it follows that 
\begin{gather*}
\frac{1}{q}Z_1
\geq  \frac{\clbi\clbiiii\ln(2) }{8\clbiii\varepsilon\ln \left (\clbiiii /(2\clbiii)\right )}.
\end{gather*}
Now using that $x/\ln(x)\geq \sqrt{x}$ for $x\geq 1$ with $x=\clbiiii /(2\clbiii) $, which is greater than $2$ by assumptions, we get that  
\begin{gather*}
     \frac{\clbiiii }{2\clbiii\ln \left (\clbiiii /(2\clbiii)\right )}\geq \sqrt{\clbiiii /(2\clbiii) }.
\end{gather*}
Thus we get 
\begin{gather*}
\frac{1}{q}Z_1
\geq  \frac{\clbi\ln(2) \sqrt{\clbiiii }}{4\sqrt{2\clbiii}\varepsilon}=\frac{c_8\clbi}{\varepsilon\sqrt{\clbiii}},
\end{gather*}
with $c_8=\ln(2) \sqrt{\clbiiii }/(4\sqrt{2}).$
\end{proof}

We now notice that the assumption of $qm/L\leq 1$ in \Cref{lowerboundz1} for a fixed $\clbi$ maybe be removed.

\begin{restatable}{Remark1}{sensthirteen}
\label{remarklowerboundstep2}
We may assume that $qm/L\leq 1$ in \cref{lowerboundz1} for a fixed $\clbi$ holds by choosing $\clbxiii$ sufficiently large.
\end{restatable}
\begin{proof}
 To see this we notice that the assumption $qm/L\leq 1$ is equivalent to $$ \frac{qm \ln \left(\ln\left (1/\delta\right)/\left(qm\right)\right)}{\clbi\ln\left (1/\delta\right)}\leq 1.$$
So if we can upper bound the left hand side by 1, we are done. 
To upper bound the left hand side we use that $\ln \left(x\right)/x$ is decreasing for $x\geq 3$ so using this fact with $x=\ln (1/\delta)/(qm)$ and $\ln (1/\delta)/(qm)$ being lower bounded by $\clbxiii$ (\cref{constantc7a}) we get that
$$ \frac{qm \ln \left(\ln (1/\delta)/(qm)\right)}{\clbi\ln (1/\delta)}\leq \frac{\ln \clbxiii}{\clbi\clbxiii},$$
which is less than 1 for sufficiently large $\clbxiii$ hence the assumption of $qm/L\leq 1$ for a fixed $\clbi$ may be removed.

\end{proof}

\begin{restatable}{Lemma1}{sensfourteen}
\label{lowerboundz1i}
Let the setting be as in \Cref{lowerboundz1} other than $m/L\leq 1$ then we have with probability $\delta^{\clbi}$ that
$$\frac{1}{q}Z_1\geq \frac{1}{q}\geq \frac{1}{\clbiii\varepsilon}.$$
\end{restatable}

\begin{proof}
Now since $1/q\geq Z_1/q$ happens if and only if $Z_1=(1/m)\sum_{j=1}^{m} b_{1,j}=1$, hence all the Bernoulli trails in the binomial being one, the above happens with probability $q^m$. This probability is less than or equal to $\left(qm/L\right)^L$ since $m/L\leq 1$ now the calculations in \Cref{lowerboundz1} for $\left(qm/(2L)\right)^L$ yields that $q^m \geq \delta^{\clbi}$. The later lower bound on $1/q$ follows from $q\leq\clbiii\varepsilon$ (\cref{constansc6a})
\end{proof}

We now show that with good probability we have that $N_1^2$ is $\Theta(\ln (1/\delta))$.
 
\begin{restatable}{Lemma1}{sensfifteen}
\label{lowerboundN1}
For $x\geq0$ we have with probability at least $1-\sqrt{1-\exp(-2x/\pi})$ that
\begin{gather*}
 N^2\geq x.
 \end{gather*}
\end{restatable}
\begin{proof}
For showing this we will us an upper bound on the error function and here by get at lower bound on the two tails of the standard normal distributions. The error function is defined as $\mathrm{erf}(x):= (2/\sqrt{\pi})\int_{0}^{x}e^{-x^2} \  dx$ and has the property that $\Phi(x)=(1+\mathrm{erf}(x/\sqrt{2}))/2$ where $\Phi$ denote the cdf of the standard normal distribution. We will use the following upper bound $\mathrm{erf}(x)<\sqrt{1-\exp(-4x^2/\pi})$ from \cite{sym12122017}. Now using the symmetry of the standard normal distribution around 0 we get
\begin{gather*}
\ppp\left[N^2\geq x\right] 
=\ppp\left[N\leq -\sqrt{x},N\geq \sqrt{x}\right] 
=2\left(1-\Phi\left(\sqrt{x}\right)\right).
\end{gather*}
Now using $\Phi(x)=(1+\mathrm{erf}(x/\sqrt{2}))/2$ we get
\begin{gather*}
\ppp\left[N^2\geq x\right]
=2\left(1-\left(1+\mathrm{erf}\left(\sqrt{x/2}\right)\right)/2\right)
=1-\mathrm{erf}\left(\sqrt{x/2}\right).
\end{gather*}
Lastly using $\mathrm{erf}(x)<\sqrt{1-\exp(-4x^2/\pi})$ we get
\begin{align*}
\ppp\left[N^2\geq x\right] \geq1- \sqrt{1-\exp\left(-2x/\pi\right)},
\end{align*}
Which concludes the proof.
\end{proof}

We will now combine \Cref{lowerboundz1}, \Cref{remarklowerboundstep2}, \Cref{lowerboundz1i} and \Cref{lowerboundN1}  to show \Cref{lowerbound1/qZN2}, recall that \Cref{lowerbound1/qZN2} is.

\begin{customlem}{\ref{lowerbound1/qZN2}}
\senseight*
\end{customlem}

\begin{proof}

Let $\clbi=1/50$ and now fix $\clbxiii$ large enough such that $qm/L\leq1$ as described in \Cref{remarklowerboundstep2} and such that $\clbxiii$ is greater than $8$. Then we have with probability $\delta^{1/50}$ by either \Cref{lowerboundz1} (and accordingly small $\clbiii$) or \Cref{lowerboundz1i}  that
$$\frac{1}{q}Z_1\geq \min\left (\frac{1}{\clbiii\varepsilon} , \frac{c_8}{50\varepsilon\sqrt{\clbiii}}\right).$$ 
We now also choose $\clbiii$ so small that the above is greater than $2\cdot5\cphdii\varepsilon^{-1}$. 

Now using $\sqrt{1-x}\leq 1-x/2$ for $x\leq 1$ and that $\delta\leq 1/4$ it follows by \Cref{lowerboundN1} that with probability $1-\sqrt{1-\exp(-\ln (1/\delta) /\pi})\geq\delta^{1/\pi}/2\geq \delta^{1/2+1/\pi}$, we have  $N_1^2\geq \ln (1/\delta)/2$.

Now since that $Z_1$ and $N_1^2$ are independent we conclude that with probability $\delta^{1/50+1/2+1/\pi}$ we have that
\begin{gather*}
\frac{1}{q}Z_1N_1^2\geq \frac{2\cdot5\ln (1/\delta)}{2\varepsilon}=\frac{5\ln (1/\delta)}{\varepsilon},
\end{gather*}

which concludes the proof of \Cref{lowerbound1/qZN2}
\end{proof}

We now restate and prove \Cref{lowerboundsumN}. 
\begin{customlem}{\ref{lowerboundsumN}}
\sensnine*
\end{customlem}

\begin{proof}
Let $X=(1/q)\sum_{i=2}^{k}Z_iN_i^2\stackrel{d}{=}(1/(mq))\sum_{i=2}^{k}b_iN_i^2$, where the $b_i$'s are binomial random variables with $m$ trails and success probability $q$, the $N_i$'s are standard normal random variables and the $b_i$'s and the $N_i$'s are all independent of each other. We now notice since the $b_iN_i^2$'s are independent and identically distributed the variance of their sum i equal to $k-1$ times the variance of $b_2N_2^2$:
\begin{gather*}
\mathrm{Var}\left(X\right)=\frac{1}{\left(mq\right)^2}\sum_{i=2}^{k}\mathrm{Var}\left(b_iN_i^2\right)=\frac{k-1}{\left(mq\right)^2}\mathrm{Var}\left(b_2N_2^2\right).
\end{gather*}
Now using the independence of $b_2$ and $N_2$ and  that the forth moment of a standard normal distribution is $3$, and that the first and second moment of a binomial random variable is respectively $mq$ and $(mq)^2+mq(1-q)$ we get that
\begin{gather*}
 \mathrm{Var}\left(b_2N_2^2\right)=\E\left[ \left(b_2N_2^2\right)^2 \right ]-E\left [\left(b_2N_2^2\right)\right ]^2=\E\left [b_2^2\right ]\E\left [N_2^4\right ]-\left(\E\left [b_2\right ]\E\left [N_2^2\right ]\right)^2\\
 =3\left(\left(mq\right)^2+mq(1-q)\right)-\left(mq\right)^2=\left(mq\right)^2\left(2+(1-q)/\left(mq\right)\right).
 \end{gather*}
 Now plugging $\mathrm{Var}(b_2N_2^2)$ back into the expression of $\mathrm{Var}\left(X\right)$, yields that 
 \begin{gather*}
 \mathrm{Var}\left(X\right)=\left(k-1\right)\left(2+(1-q)/\left(mq\right)\right).
 \end{gather*}
 Now using that $\E\left [X\right ]=(k-1)$, the above calculation of the variance of $X$ and Chebyshev-Cantelli's inequality $\ppp\left[Y-\mathrm{E}[Y] \leq -t\right]  \leq \operatorname{Var}(Y) /\left(\operatorname{Var}(Y)+t^{2}\right)$ which holds for $ t>0$, yields that

\begin{gather*}
 \ppp\left[\sum_{i=2}^k \frac{1}{q}Z_iN_i^2\leq \left(1-3\varepsilon\right)\left(k-1\right)\right] 
 \leq \frac{\left(k-1\right)\left(2+\left(1-q\right)/\left(mq\right)\right)}{\left(k-1\right)\left(2+\left(1-q\right)/\left(mq\right)\right)+\left(3\varepsilon\left(k-1\right)\right)^2} \\ \leq \frac{\left(2+\left(1-q\right)/\left(mq\right)\right)}{\left(2+\left(1-q\right)/\left(mq\right)\right)+\left(3\varepsilon\right)^2\left(k-1\right)}.
\end{gather*}
Since $y\rightarrow y/(y+a)$ is increasing in $y$ for $a,y>0$, it now follows using this with $a=(3\varepsilon)^3(k-1)$ and $y=2+(1-q)/(mq)\leq 2+1=3$, where we have used that $1/(mq)\leq1$ by the comment under \cref{boundonration}, we get that
\begin{gather*}
 \ppp\left[\sum_{i=2}^k \frac{1}{q}Z_iN_i^2\leq \left(1-3\varepsilon\right)\left(k-1\right)\right] 
 \leq \frac{3}{3+\left(3\varepsilon\right)^2\left(k-1\right)}
 \end{gather*}
 Lastly using that $k=\ln(1/\delta)/\eps^2,$ $\eps\leq 1/4$ and $\delta\leq 1/8$ we get $\varepsilon^2(k-1)=\ln(1/\delta)-\varepsilon^2\geq 2$, and we conclude that
 \begin{gather*}
 \ppp\left[\sum_{i=2}^k \frac{1}{q}Z_iN_i^2\leq \left(1-3\varepsilon\right)\left(k-1\right)\right] \leq
\frac{3}{3+18}\leq 1-\left(1/8\right)^{1/8}\leq 1- \delta^{1/8},
 \end{gather*}
 which ends the proof.
\end{proof}

\bibliographystyle{abbrv}
\bibliography{bibli}

\end{document}